\documentclass{revtex4-1}
\usepackage[T1]{fontenc}
\usepackage{hyperref}
\usepackage{color}

\usepackage{graphicx}

\usepackage{amssymb,amsfonts,amsmath,amsthm}
\usepackage{dsfont}
\DeclareMathAlphabet{\varmathbb}{U}{bbold}{m}{n}

\DeclareMathOperator*{\tr}{\mathrm{tr}}

\newcommand{\Exp}{\mathbb{E}}

\newcommand{\beq}{\begin{equation}} 
\newcommand{\eeq}{\end{equation}}
\newcommand{\bem}{\begin{multline}}
\newcommand{\bes}{\begin{split}} \newcommand{\ees}{\end{split}} 
\newcommand{\bea}{\begin{eqnarray}} \newcommand{\eea}{\end{eqnarray}}
\def\tq{\widetilde{q}}
\def\tp{\widetilde{p}}
\def\tr{\widetilde{r}}
\def\dd{{\rm d}}
\def\la{\langle}
\def\ra{\rangle}
\def\Z{{\cal Z}}

\theoremstyle{plain}
\newtheorem{thm}{\protect\theoremname}
\theoremstyle{remark}
\newtheorem*{rem*}{\protect\remarkname}
\theoremstyle{plain}
\newtheorem{cor}[thm]{\protect\corollaryname}
\theoremstyle{remark}

\providecommand{\corollaryname}{Corollary}
\providecommand{\remarkname}{Remark}
\providecommand{\theoremname}{Theorem}

\begin{document}



\title{Network dismantling}




\author{
Alfredo Braunstein}
  \affiliation{Politecnico di Torino, Torino, Italy}
  \affiliation{Human Genetics Foundation, Torino, Italy}
  \affiliation{Collegio Carlo Alberto, Moncalieri, Italy}
\author{
Luca Dall'Asta}
  \affiliation{Politecnico di Torino, Torino, Italy}
  \affiliation{Collegio Carlo Alberto, Moncalieri, Italy}
\author{
Guilhem Semerjian}
  \affiliation{LPTENS, Ecole Normale Sup\'erieure, PSL Research University, Sorbonne Universit\'es, UPMC Univ Paris 06, CNRS UMR 8549, 24 Rue Lhomond, 75005
 Paris, France}
\author{
Lenka Zdeborov\'a}
  \affiliation{Institut de Physique Th\'eorique, CNRS, CEA and Universit\'e Paris-Saclay, Gif-sur-Yvette, France.}




\begin{abstract} 
We study the network dismantling problem, which consists in 
determining a minimal set of vertices whose removal leaves the network 
broken into connected components of sub-extensive size. 
For a large class of random graphs,
this problem is tightly connected to the decycling problem (the
removal of vertices leaving the graph acyclic). Exploiting this
connection and recent works on epidemic spreading we present precise predictions for the minimal size of a
dismantling set in a large random graph with a prescribed
(light-tailed) degree distribution. Building on the statistical
mechanics perspective we propose a three-stage Min-Sum algorithm for efficiently
dismantling networks, including heavy-tailed ones for which the
dismantling and decycling problems are not equivalent. We also provide
further insights into the dismantling problem concluding that it is
an intrinsically collective problem and that optimal dismantling sets cannot
be viewed as a collection of individually well performing nodes. 
\end{abstract}

\maketitle








\section{Introduction}

A network (a graph $G$ in the discrete mathematics language) is a set $V$ of $N$ 
entities called nodes (or vertices), along with a set $E$ of edges connecting some 
pairs of nodes.
In a simplified way, networks are used to describe numerous
systems in very diverse fields, ranging from social sciences 
to information technology or biological systems, 
see~\cite{Boccaletti_review,Barrat_book} for reviews. 
Several crucial questions in the context of network studies concern the
modifications of the properties of a graph when a subset $S$ of its nodes 
is selected and treated in a specific way. For instance: 
How much does the size of the largest connected component of 
the graph decreases if the vertices in~$S$ (along with their adjacent edges) 
are removed? Do the cycles survive this removal? What is the 
outcome of the epidemic spreading if the vertices in~$S$ are initially 
contaminated, constituting the seed of the epidemic? On the contrary, what is 
the influence of a vaccination of nodes in~$S$ preventing them from transmitting 
the epidemic? 
It is relatively easy to answer these questions when the set $S$ is 
chosen randomly, with each vertex being selected with some probability independently. Classical percolation theory is 
nothing but the study of the connected components of a graph in which some
vertices have been removed in this way.

A much more interesting case is when the set $S$ can be 
chosen in some optimal way. Indeed, in all applications sketched above it is
reasonable to assign some cost to the inclusion of a vertex in~$S$: vaccination
has a socioeconomic price, incentives must be paid to customers to
convince them to adopt a new product in a viral marketing campaign, 
incapacitating a computer during a cyber attack requires resources.
Thus, one faces a combinatorial optimization problem: the minimization
of the cost of $S$ under a constraint on its effect on the graph. These 
problems thus exhibit both static and dynamic features, the former referring to the combinatorial 
optimization aspect, the latter to the definition of the cost function itself 
through a dynamical process.

In this paper we focus on the the existence of a giant
component in a network, that is the largest component containing a positive fraction of the vertices (in the
$N\to\infty$ limit). On the one hand, the existence of a giant component is often 
necessary for the network to fulfill its function (e.g. to deliver
electricity, information bits or ensure possibility of
transportation). An adversary might be able to destroy a set of nodes
with the goal of destroying this functionality. 
It is thus important to understand what is an optimal attack strategy,
possibly as a first step in the design of optimal defense strategies. 
On the other hand, a giant component can propagate an epidemic to a
large fraction of a population of nodes. Interpreting the removal of
nodes as the vaccination of individuals who cannot transmit the
epidemic anymore, destroying the giant component can be seen as an extreme way of organizing a vaccination campaign
\cite{pastor-satorras_immunization_2002,cohen_efficient_2003} by
confining the contagion to small connected components (less drastic strategies can be devised using
specific information about the epidemic propagation
model~\cite{britton_graphs_2007,vaccination_Torino}). 
Another related application is {\it influence maximization} as studied in many previous
works~\cite{Kempe,Chen,Dreyer09}. In particular, optimal destruction of the giant
component is equivalent to selection of the smallest 
set of initially informed nodes needed to spread the information into
the whole network under a special case of the commonly considered 
model for information spreading~\cite{Kempe,Chen,Dreyer09}. 

To define the main subject of this paper more formally, following~\cite{janson_dismantling_2008}, we call $S$
a {\it $C$-dismantling} set if its removal yields a graph whose largest
component has size (in terms of its number of nodes) at most $C$. The 
{\it $C$-dismantling number} of a graph is the minimal size of such a set.
When the value of $C$ is either clear from the context or not important for 
the given claim, we will simply talk about {\it dismantling}. Typically, the 
size of the largest component is a finite fraction of the total
number of nodes $N$. To formalize the notion of {\it destroying the giant
component} we will consider the bound $C$ on the size of the connected
components of the dismantled network to be such that
$1 \ll C \ll N$. 
It should be noted that we defined dismantling in terms of node
removal, it could be rephrased in terms of
edge-removal~\cite{BrMoYa13}, which turns out to be a much easier problem.
The dismantling problem is also referred to
as {\it fragmentability of graphs} in graph theory literature~\cite{edwards1994new,edwards_fragmentability_2001,edwards_planarization_2008}, and as
{\it optimal percolation} in~\cite{morone_influence_2015}.

Determining whether the 
$C$-dismantling number of a graph is smaller 
than some constant is an NP-complete decision problem
(for a proof see Appendix A). 
The concept of NP-completeness concerns the worst-case difficulty of
the problem. 
The questions we address in the present paper are instead the following: What is the
dismantling number on some representative class of graphs, in our case
random graphs? What are the best heuristic algorithms and how does
their performance compare to the optimum, and how do they perform on
benchmarks of real-world graphs?
Simple heuristic algorithms for the dismantling problem were considered in previous
works~\cite{albert_error_2000,callaway_network_2000,cohen_breakdown_2001},
where the choice of the nodes
to be included in the dismantling set was based on their degrees (favoring
the inclusion of the most connected vertices), or some measure of their
centrality. More recently, a heuristic for the dismantling problem
has been presented in~\cite{morone_influence_2015} under the name 
``collective'' influence, in which the inclusion of a node is decided according 
to a combination of its degree and the degrees of the nodes in a local 
neighborhood around it. Ref.~\cite{morone_influence_2015} also
attempts to estimate the dismantling number on random graphs. 

\bigskip

\vspace{-0.5cm}
\section{Our main contribution}

In this paper we provide a detailed study of the dismantling
problem, with both analytical and algorithmic outcomes. 
We present very accurate estimates of
the dismantling number for large random networks, building on a
connection with the decycling problem (in which one seeks a subset of
nodes whose removal leaves the graph acyclic, also an NP-complete problem~\cite{Karp1972})
and on recent studies of
optimal spreading~\cite{Torino1,Torino2,fvs_Zhou1,Guilhem14}. Our
results are the one-step replica symmetry broken estimate of the
ground state of the corresponding optimization problem. 

On the computational side, we introduce a very efficient
algorithm that outperforms considerably state-of-the-art algorithms
for solving the dismantling problem. We demonstrate its efficiency 
and closeness to optimality both on random graphs and on real world networks.
The goal of our paper is closely related to the one 
of~\cite{morone_influence_2015}, we present an assessment 
of the results reported therein, on random as well as on real world networks. 

Our dismantling algorithm, which has been inspired by the theoretical
insight gained on random graphs, is composed of three stages:
\begin{itemize}
   \item[(1)] {\bf Min-sum message passing for decycling.} This is the core
     of the algorithm, employing a variant of a message-passing
     algorithm developed in \cite{Torino1,Torino2}.
     A related but different message-passing algorithm 
     was developed for decycling in \cite{fvs_Zhou1}, and later applied to
     dismantling in \cite{Zhou_dis}, it performs comparably to ours.
   \item[(2)] {\bf Tree breaking.} Once all cycles are broken, some of
     the tree components may still be larger than the desired threshold $C$. 
     We break them into small components removing a fraction of nodes that
     vanishes in the large size limit. This can be done in time $O(N\log{N})$ 
     by an efficient greedy procedure (detailed in Appendix~\ref{sec_tree_breaking}).
   \item[(3)] {\bf Greedy reintroduction of cycles.} As explained below the strategy
     of first decycling a graph before dismantling it is the optimal one for
     graphs that contain few short cycles (a typical property of light-tailed
     random graphs). For graphs with many short cycles we improve considerably
     the efficiency of our algorithm by reinserting greedily some nodes
     that close cycles without increasing too much the size of the largest 
     component. 
\end{itemize}

The dismantling problem, as is often the case in combinatorial optimization,
exhibits a very large number of (quasi-)optimal solutions. We characterize
the diversity of these degenerate minimal dismantling sets by a detailed
statistical analysis, computing in particular the frequency of appearance of
each node in the quasi-optimal solutions, and conclude that dismantling is 
an intrinsically collective phenomenon that results from a correlated choice 
of a finite fraction of nodes. It thus makes much more sense to think in terms 
of {\it good dismantling sets} as a whole and not about individual nodes as the 
{\it optimal influencers/spreaders}~\cite{morone_influence_2015}. 
We further study the correlation between the degree of a node and its
importance for dismantling, exploiting a natural variant of our algorithm 
in which the dismantling set is required to avoid some marked nodes.
This allows us to show that each of the low degree nodes can be
replaced by other nodes without increasing significantly 
the size of the dismantling set. Contrary to claims 
in~\cite{morone_influence_2015} we do not confirm any particular importance 
of weak-nodes, apart from the obvious fact that the set of highest degree 
nodes is not a good dismantling set.  

To give a quantitative idea of our algorithmic contribution, we state 
two representative examples of the kind of improvement we obtain
with the above algorithm with respect to the 
state-of-the-art~\cite{morone_influence_2015}. 
\begin{itemize}
\item
In an Erd\H os-R\'enyi (ER) random graph of average degree 3.5 and 
size $N=5^7$ we found $C=1000$-dismantling sets 
removing $17.8\%$ of the nodes, whereas the best known method (adaptive 
eigenvalue centrality for this case) removes $20.2\%$ of the nodes, 
and the adaptive ``collective'' influence (CI) method 
of~\cite{morone_influence_2015} removes $20.6\%$ of the nodes. 
Hence, we provide a $13\%$ improvement over the state of the art. Our 
theoretical analysis estimates the optimum dismantling number
to be around $17.5\%$ of the nodes, thus the algorithm is extremely close to 
optimal in this case. 
\item Our algorithm managed to dismantle the Twitter network studied 
in~\cite{morone_influence_2015} (with $532,000$ nodes) into components smaller 
than $C=1000$ using only $3.4\%$ of the nodes, whereas the CI heuristics of
\cite{morone_influence_2015} needs $5.6\%$ of the nodes. Here we thus
provide a $60\%$ improvement over the state-of-the-art.  
\end{itemize}

Not only does our algorithm demonstrate beyond state-of-the-art
performance, but it is also computationally efficient. Its core
part runs in linear time over the number of edges, allowing 
us to easily dismantle networks with tens of millions of nodes. 

\vspace{-0.5cm}
\section{The relation between dismantling and decycling}

We begin our discussion by clarifying the relation between the
dismantling and decycling problems. While the argument below can be found
in~\cite{janson_dismantling_2008}, we reproduce
it here in a simplified fashion. The decycling number (or more precisely 
fraction)
$\theta_{\rm dec}(G)$ of $G$ is the minimal fraction of vertices that
have to be removed to make the graph acyclic.
We define similarly the dismantling number $\theta_{\rm dis}(G,C)$ of a graph
$G$ as the minimal fraction of vertices that have to be removed to
make the size of the largest component of the remaining graph smaller
than a constant $C$. 

For random graphs with degree distribution $q=\{q_k\}_{k \ge 0}$, in the large size limit, the parameters
$\theta_{\rm dec}$ and $\theta_{\rm dis}$ will enjoy concentration
(self-averaging) properties, we shall thus write their typical values as 
\bea
\theta_{\rm dec}(q) &=& \lim_{N \to \infty} \mathbb{E}[ \theta_{\rm
  dec}(G)] \ ,  \\
\theta_{\rm dis}(q) &=&  \lim_{C\to \infty} \lim_{N \to \infty} \mathbb{E}[ \theta_{\rm dis}(G,C)] \ ,
\eea
where $\mathbb{E}[\bullet]$ denotes an average over the random graph ensemble.
For the dismantling number we allow the connected components after the removal of a
dismantling set to be large but sub-extensive because of the order of
limits. It was proven in~\cite{janson_dismantling_2008} that
for some families of random graphs an
equivalent definition is $\lim_{\epsilon \to 0} \lim_{N \to \infty}
\mathbb{E}[ \theta_{\rm dis}(G,\epsilon N)]$, i.e. connected
components are allowed to be extensive but with a vanishing intensive size.

The crucial point for the relation between dismantling and decycling is that trees (or more generically
forests) can be efficiently dismantled. It was proven
in~\cite{janson_dismantling_2008} that 
$\theta_{\rm dis}(G,C)\le {1/(C+1)}$ whenever $G$ is a forest. This
means that the fraction of vertices to be removed from a forest
to dismantle it into components of size $C$ goes to zero when $C$
grows.

This observation brings us to the following two claims concerning the
dismantling and decycling numbers for random graphs with degree
distribution $q$: (i) for any degree distribution, 
$\theta_{\rm dis}(q) \le \theta_{\rm dec}(q)$; (ii) if $q$ also
admits a second moment (we shall call $q$ light-tailed when this is the case) then there is actually an equality
between these two parameters, $\theta_{\rm dis}(q) = \theta_{\rm dec}(q)$.

The first claim follows directly from the above observation on the
decycling number of forests. Once a decycling set $S$ of $G$ has been
found one can add to $S$ additional vertices to turn it into a
$C$-dismantling set, the additional cost being bounded as $\theta_{\rm
dis}(G,C) \le \theta_{\rm dec}(G) + {1}/{(C+1)}$. Taking averages of
this bound and the limit $C\to\infty$ after $N\to\infty$ yields
directly (i).

To justify our second claim, we consider a $C$-dismantling set $S$ of a graph $G$. To turn $S$ into a
decycling set we need to add additional vertices in order to break the
cycles that might exist in $G \setminus S$. The lengths of these cycles are certainly
smaller than $C$, and removing at most one vertex per cycle is enough
to break them. We can thus write $\theta_{\rm dec}(G) \le \theta_{\rm
  dis}(G,C)+{n_C(G)}/{N}$, with $n_C(G)$ denoting the number of cycles
of $G$ of length at most $C$. We recall that the existence of a second moment of $q$ implies that $n_C(G)$ remains bounded when $N\to\infty$ with
$C$ fixed. Considering the limit $N\to\infty$ and
property (i), property (ii) follows. 

\vspace{-0.5cm}
\section{Network decycling}

In this section, we shall explain the results on the decycling number
of random graphs we obtained via statistical mechanics methods, and
how they can be exploited to build an efficient heuristic algorithm
for decycling arbitrary graphs. 

\subsection{Testing the presence of cycles in a graph}

The 2-core of a graph $G$ is its largest subgraph of minimal degree 2; 
it can be constructed by iteratively removing isolated nodes and leaves 
(vertices of degree 1) until either all vertices have been
removed or all remaining vertices have degree at least 2. It is easy to see 
that a graph contains cycles if and only if its 2-core is non-empty.
To decide if a subset $S$ is decycling or not we remove the nodes in 
$S$ and perform this leaf
removal on the reduced graph. To formalize this we introduce binary variables
$x_i^t(S) \in \{0,1\}$ on each vertex $i \in V$ of the graph, $t$ being a
discrete time index. At the starting time $t=0$, one marks the
initially removed vertices by setting $x_i^0(S)=1$ if $i \in S$, 0
otherwise, and let the $x$ variables evolve in time according to
\begin{equation}
x_i^{t+1}(S)= \begin{cases} 1 & \text{if} \ x_i^t(S)=1\, , \\
\mathbb{I}\left[\underset{j\in\partial i}{\sum}(1-x_j^t(S))\leq 1 \right] &
\text{if} \ x_i^t(S)=0\, ,
\end{cases}
\label{eq:dynamics}
\end{equation}
where $\partial i=\{j:(ij)\in E\}$ denotes the local neighborhood of
vertex $i$ and $\mathbb{I}$ denotes the indicator function, that is 1
if its argument is true and 0 otherwise. One can check that the
$x_i$'s are monotonous in time (they can only switch from 0 to 1),
hence they admit a limit $x_i^*(S)$ when $t \to \infty$. At this fixed point $x_i^*(S)=0$ if and only if $i$ is
in the 2-core of $G \setminus S$, hence the sufficient and necessary
condition for $S$ to be a decycling set of $G$ is $x_i^*(S)=1$ for all
vertices $i$.

Note that the leaf-removal procedure can be equivalently viewed as
a particular case of the linear threshold model of epidemic
propagation or of information spreading. By
calling a removed vertex infected (or informed), one sees that the
infection (or information) of node $i$ occurs whenever the number of
its infected (or informed)
neighbors reaches its degree minus one. This equivalence, which was
already exploited in~\cite{Guilhem14,morone_influence_2015}, allows us to build on
previous works on minimal contagious
sets~\cite{Torino1,Torino2,Guilhem14} and on influence maximization ~\cite{Kempe,Chen,Dreyer09}. 

\subsection{Optimizing the size of decycling sets}
From the point of view of statistical mechanics, it is natural to
introduce the following probability distribution over the subsets 
$S$ in order to find the optimal decycling sets of a given graph
\begin{equation}\label{eq:pmeasure}
\widehat{\eta}(S) = \frac{1}{Z(\mu)} e^{ \mu |S|} \prod_{i \in V}
\mathbb{I}[x_i^*(S)=1] \ ,
\end{equation}
where $|S|$ denotes the number of vertices in $S$, $\mu$ is a real
parameter to be interpreted as a chemical potential (or an inverse temperature), and the
partition function $Z(\mu)$ normalizes this probability
distribution. 
From the preceding discussion, this measure gives a
positive probability only to decycling sets, and their minimal size
can be obtained as the ground state energy in the zero-temperature
limit:
\begin{equation}
\theta_{\rm dec}(G) = \frac{1}{N} \lim_{\mu \to -\infty}
\frac{1}{\mu} \ln Z(\mu) \ .
\end{equation}
The computation of this partition function remains at this
point a difficult problem, in particular the variables $x_i^*$ depend on the choice of $S$
in a non-local way. One can get around this difficulty in the
following way: as the evolution of $x_i^t$ is monotonous in time it
can be completely described by a single integer, 
$t_i(S)=\min\{t : x_i^t(S)=1\} $, the time at which $i$ is removed in
the parallel evolution described above. Note that $t_i(S)=0$ if and
only if $i \in S$, $t_i(S)>0$ otherwise. We use the natural convention
$\min\emptyset=\infty$, hence the nodes $i$ in the 2-core of
$G\setminus S$ are precisely those with an infinite removal time $t_i(S)=\infty$.
The crucial advantage of this equivalent representation in terms of the
activation times is its locality along the graph. Indeed, the dynamical evolution rule \eqref{eq:dynamics} can be
rephrased as static equations linking the times $t_i$ on neighboring
vertices: 
\begin{equation}
t_i (S)= \begin{cases} 0 & \text{if} \ \ i \in S\, , \\
\phi_i(\{t_j\}_{j \in \partial i})
& \text{if}  \ \ i \in V
\setminus S \ ,
\end{cases}  
\label{eq:dynamics-trajectories}
\eeq
\beq
\text{with} \ \ \phi_i(\{t_j\}_{j \in \partial i}) = 1+ {\rm max}_2(\{t_j(S)\}_{j \in \partial i} )
\end{equation}
where we denote $\max_2$ the second largest of the arguments
(reordering them as $t_1 \ge t_2 \ge \dots \ge t_n$ one defines
$\max_2(t_1,\dots,t_n)=t_2$). In the leaf-removal procedure, one
vertex is removed in the first step following the time at which all
but one of its neighbors have been removed, making it a leaf.
The set of equations \eqref{eq:dynamics-trajectories} admits a unique
solution for each $S$, hence the partition function can be rewritten as:
\beq
Z(\mu) = \sum_{\{t_i\}} e^{ \mu \sum_i \psi_i(t_i)} 
\prod_{i \in V} \mathbb{I}[t_i < \infty] 
\prod_{i \in V}  \Phi(t_i,\{t_j\}_{j\in \partial i})\ , \label{eq:Z_local} 
\eeq
with $\psi_{i}\left(t_{i}\right)=\mathbb{I}\left[t_{i}=0\right]$, and
$\Phi(t_i,\{t_j\}_{j\in \partial i})=\mathbb{I}[t_i =0] + \mathbb{I}[t_i = \phi_i(\{t_j\}_{j \in \partial i})] $. We have thus obtained an exact
representation of the generating function counting the number of
decycling sets according to their size as a statistical mechanics model of
variables (the $t_i$'s) interacting locally along the graph $G$. We transformed the non-equilibrium problem of
leaf removal into an equilibrium problem where the times of removal play the
role of the static variables. Note 
that Ref.~\cite{fvs_Zhou1}, which also estimates the decycling number, uses a simpler, but approximate,
representation, where one cycle may remain in every connected
component, and the correspondence between microscopic configurations and sets of removed vertices is many to one. The
domain of the variables $t_i$ should include all integers between $0$ and
the diameter of $G$, and the additional $\infty$ value. For practical
reasons, in the rest of this paper we restrict this set to $\{0,1,\dots,T,\infty\}$, where
$T$ is a fixed parameter, and project all $t_i$'s
greater than $T$ to $\infty$. This means that we
require $G\setminus S$ not only to be acyclic, but that its connected
components are trees of diameter at most $T$. For large enough values
of $T$ this additional restriction is inconsequential.

The exact computation of the partition function \eqref{eq:Z_local} for
an arbitrary graph remains an NP-hard problem. However, if
$G$ is a sparse random graph the large size limit of its free-energy
density $\ln Z(\mu)/(N\mu)$ can be computed by the cavity
method~\cite{cavity,MeMo_book}. The latter has been developed for
statistical mechanics models on locally tree-like graphs, such as 
light-tailed random graphs, for which the exactness of the cavity method
has been proven mathematically on several problems. 
The starting point of the method is based on the fact that light-tailed random
graphs converge locally to trees in their large size limit, hence
models defined on them can be treated with belief propagation (BP, also
called Bethe Peierls approximation in statistical mechanics).
In BP, a partition function akin to \eqref{eq:Z_local} is
computed via the exchange of {\it messages} between neighboring nodes.
In the present case, where an interaction in \eqref{eq:Z_local}
includes node $i$ and all its neighbors $j\in \partial i$, we
obtain a tree-like representation if we let pairs of variables
$t_i,t_j$ live on the edges and add consistency constraints on the
nodes. The BP message $\eta_{ij}(t_i,t_j)$ from $i$ to
$j\in \partial i$ is then a function of both the activation times $t_i$ and
$t_j$. This message is interpreted as the marginal probability law of
the local variables $t_i$ and $t_j$ in an amputated (cavity) graph in
which the interaction between $i$ and $j$ has been removed. 
Thanks to the locally tree-like character of the graph, some correlation-decay properties are
verified and allow a node's incoming messages to be treated as
independent. Under this assumption, the iterative BP equations~\cite{Torino1,Torino2,Guilhem14},
for decycling are written as
\begin{equation}
\eta_{ij}(t_i,t_j) \propto \! \! \! \! \! \! \sum_{\{t_k\}_{k\in\partial i\setminus
    j}} \! \! \!  \! \! \!  e^{ \mu \psi_i(t_i)} 
\Phi(t_i,\{t_k\}_{k\in \partial i})
\! \! \!   \prod_{k\in\partial i\setminus
  j} \! \! \! \eta_{ki}(t_k,t_i) \ ,  \label{eq:BPeq}
\end{equation}
the $\propto$ symbol includes a multiplicative normalization constant.
The free-energy can then be computed as a sum of local
contributions depending on the messages solution of the BP
equations.

Better parametrizations with a number of real values per message that
scales linearly with $T$ (rather that quadratically) can be devised
\cite{Torino2, Guilhem14}. A parametrization with $2T$ real values per
message was introduced in \cite{Guilhem14} and was employed to obtain
improved results for the minimum decycling set on regular random
graphs by extending the cavity method to the so-called first level of
the replica symmetry breaking (1RSB) scheme.  
The extension of this calculation to random graphs with
arbitrary light-tailed degree distributions is reported in the Appendix~\ref{sec_analytic}
(along with expansions close to the percolation threshold and at large
degrees, and a lower bound on $\theta_{\rm dec}$ valid for all graphs). The
1RSB predictions for the decycling 
 fraction $\theta_{\rm dec}(d)$ of Erd\H os-R\'enyi
 random graphs with average degree $d$, obtained solving numerically
 the corresponding equations and extrapolating the 
results in the large $T$ limit, are presented for a few values of $d$
in Table~\ref{tab:ER1RSB}. 

\begin{table*}
\begin{tabular}{ccc}
$d$ & $\theta_{\rm dec}(d)$ & $\theta^{\rm MS}_{\rm dec}(d)$ \\
\hline
\hline
1.5 & 0.0125 & 0.0135\\
\hline
2.5 & 0.0912 & 0.0936\\
\hline
3.5 & 0.1753 & 0.1782\\
\hline
5 & 0.2789 & 0.2823\\
\end{tabular}
\caption{The (1RSB) cavity predictions for the decycling number of
  Erd\"{o}s-R\'enyi random graphs of average degree $d$, and the decycling number reached by the Min-Sum algorithm on graphs of size $N=10^7$ nodes.\label{tab:ER1RSB}}
\end{table*}

\subsection{Min-Sum algorithm for the decycling problem}
We turn now to the description of our heuristic algorithm for finding
decycling sets of the smallest possible size. The above analysis shows
the equivalence of this problem with the minimization of the cost
function $\sum_i \psi_i(t_i)$ over the feasible configurations of the
activation times $\{t_i\}\in\{0,\dots,T\}^V$, where feasible means
that, for all vertices $i$, either $t_i=0$ (then $i$ is included in the decycling set
$S$) or if $t_i>0$ it obeys the constraint 
$t_i = 1+ {\rm max}_2(\{t_j\}_{j \in \partial i} )$. Since this
minimization is NP-hard, we formulate a heuristic strategy in the
following manner. We first consider a slightly modified cost function with
$\psi_{i}\left(t_{i}\right)=\mathbb{I}\left[t_{i}=0\right] +
\varepsilon_i(t_i)$, where $\varepsilon_i(t_i)$ is a randomly chosen
infinitesimally small cost associated with the removal of node $i$ at
time $t_i$. The minimum $\{t_i^*\}$ of this cost function is now unique with
probability 1, and can be constructed as $t_i^*={\rm argmin} \, h_i(t_i)$,
where the {\it field} $h_i(t_i)$ is the minimum cost among the feasible
configurations with a prescribed value for the removal time $t_i$ of
site $i$. From the solution of this combinatorial optimization problem
we construct one of the minimal decycling sets $S$ by including vertex
$i$ in $S$ if and only if $t_i^*=0$. It remains now to find a good
approximation for $h_i$; we compute it by the Min-Sum (MS) algorithm, which
corresponds to the $\mu\to -\infty$ limit of BP and is similarly based on the exchange of messages $h_{ij}(t_i,t_j)$ between neighboring
vertices, an analog of $\eta_{ij}(t_i,t_j)$, but interpreted as a minimal cost
instead of a probability. We defer to the Appendix~\ref{sec_MS} for a full derivation and
implementation details, stating here only the final equations. For $T \ge t_i>0$:
\begin{subequations}
\begin{eqnarray}
  h_{i}\left(t_{i}\right) & = & \psi_{i}\left(t_{i}\right) +
  \sum_{k\in\partial i} L_{ki}(t_i)+M_i\left(t_{i}\right)\, ,\label{eq:field1b-main}\\
  h_{i}\left(0\right) & = & \psi_{i}\left(0\right) +
  \sum_{k\in\partial i} R_{ki}(0)\, ,\label{eq:field2b-main}
\end{eqnarray}
\end{subequations}
where 
\beq
M_i(t_i)=\min\{0,\min_{k\in\partial i}
\{R_{ki}(t_i)-L_{ki}(t_i)\}\} \ , 
\eeq
and $L_{ij}$, $R_{ij}$, $M_{ij}$,
$h^0_{ij}$, and $h^1_{ij}$ form a solution of the following system of
fixed-point equations for messages defined on each directed edge $i
\to j$ of the graph: 
\begin{subequations}
\begin{eqnarray}
L_{ki}\left(t_{i}\right) & = &
\min_{t_{k}<t_{i}}h_{ki}^{0}\left(t_{k}\right) \, ,\label{eq:ms_L} \\
R_{ki}\left(t_{i}\right) & = & \min\left\{
  h_{ki}^{0}\left(t_{i}\right),\min_{t_{k}>t_{i}}h_{ki}^{1}\left(t_{k}\right)\right\}\, ,  \label{eq:ms_R}
\\
M_{ij}\left(t_{i}\right) & = & \min\left\{0,\min_{k\in\partial
    i\setminus j} \{R_{ki}(t_i)-L_{ki}(t_i)\}\right\}  \, ,\label{eq:ms_M}\\
h_{ij}^{0}\left(t_{i}\right) & \propto & \psi_{i}\left(t_{i}\right)+
\sum_{k\in\partial i\setminus j} L_{ki}(t_i) \, , \label{eq:ms_h0t}\\
h_{ij}^{1}\left(t_{i}\right) & \propto & \psi_{i}\left(t_{i}\right)+
\sum_{k\in\partial i\setminus j} L_{ki}(t_i)+M_{ij}\left(t_{i}\right)
\, ,\label{eq:ms_h1t}\\
h_{ij}^{0}\left(0\right) & \propto & \psi_{i}\left(0\right) +
\sum_{k\in\partial i\setminus j} R_{ki}(0) \, , \label{eq:ms_h00}
\end{eqnarray}
\end{subequations}
where $\propto$ includes now an additive normalization constant.
An intuitive interpretation of all these quantities and equations is
provided in the Appendix~\ref{sec_MS}, let us only mention at this point that the message
$h_{ij}^{0}\left(t_{i}\right)$ (resp. $h_{ij}^{1}\left(t_{i}\right)$)
is the minimum feasible cost on the connected component of $i$ in
$G\setminus j$, under the condition that $i$ is removed at time $t_i$
in the original graph assuming that $j$ is not removed yet
(resp. assuming that $j$ is already removed from $G$). 

This system can be solved efficiently by iteration. The computation of one iteration
takes $O\left(\left|E\right|T\right)$ elementary ($+$, $-$, $\times$,
$\min$) operations, where $|E|$ denotes the number of edges of the graph, and
a relatively small number of iterations are usually sufficient to
reach convergence.  
In principle one should take the cutoff $T$ on the removal
times to be greater than $N$ in order to solve the decycling problem,
we found however that using large but finite values of $T$ (i.e. constraining
the diameter of the tree components after the node removal) did not
increase extensively the size of the decycling set; in the simulations
presented below we used $T=35$. 
Note that our algorithm is very flexible and many variations can be
implemented by appropriate modifications of the cost function. 
For example, we exploited the possibility to forbid the removal of
certain marked nodes $i$ by setting $\psi_i(t_i=0)=\infty$ for them.

\begin{figure}[htb]
\begin{center}
\includegraphics[width=0.5\columnwidth]{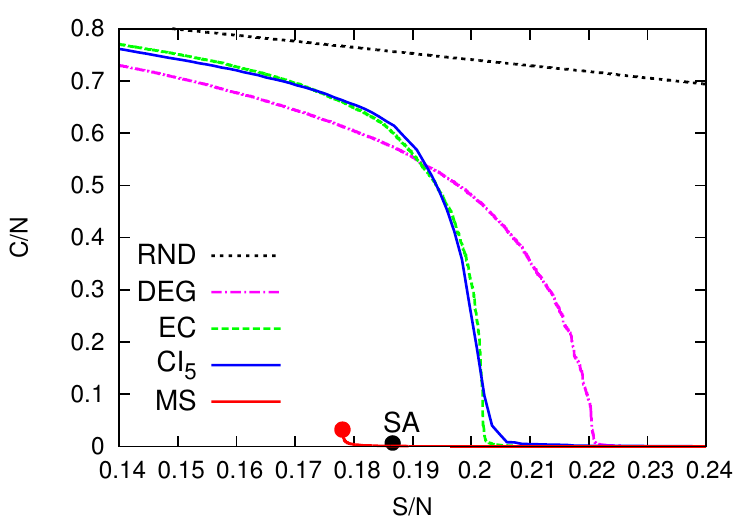}
\end{center}
\caption{Fraction of nodes in the largest component as a function of
  the fraction of removed nodes for an Erd\"{o}s-R\'enyi random graph
  of average degree $d=3.5$ and size $N=78,125$. We compare the result
  of our Min-Sum algorithm (MS) to random node removal (RND),  adaptive
  largest degree (DEG), adaptive
  eigenvector centrality (EC), adaptive ``collective'' influence
  centrality (CI) and Simulated Annealing (SA). \label{fig:threshold}} 
\end{figure}

\section{Results for dismantling}

\subsection{Results on random graphs} 
The outcome of our algorithm applied to an Erd\H os-R\'enyi random
graph of average degree $3.5$ is presented in
Figure~\ref{fig:threshold}. Here, the red point corresponds to the output
of its first stage (decycling with MS) which yields, after the removal
of a fraction $0.1781$ of the nodes, an acyclic graph whose
largest components contain a fraction $0.032$ of the vertices. 
The red line corresponds to the second stage, which
further reduces the size of the largest component by greedily breaking
the remaining trees. We compare to Simulated
Annealing (SA, black point) as well as to several incremental algorithms that successively
remove the nodes with the highest {\it scores}, where the score of a vertex is
a measure of its centrality. Besides a trivial function which gives
the same score to all vertices (hence removing the vertices in random
order, RND), and the score of a vertex equal to its degree (DEG), we used the
eigenvector centrality measure (EC) and the recently proposed Collective
Influence (CI) measure~\cite{morone_influence_2015}. We used all these
heuristics in an {\it adaptive} way, recomputing the scores after each
removal. Further details on all these algorithms can be found in Appendix~\ref{sec_otheralgo}. 

We see from the figure that the MS algorithm outperforms the others by a
considerable margin: it dismantles the graph using $13\%$ fewer
nodes than the CI method. The Monte Carlo based
SA algorithm performs rather well, but is considerably slower than all the others. 

\begin{figure}[htb]
\begin{center}
\includegraphics[width=0.5\columnwidth]{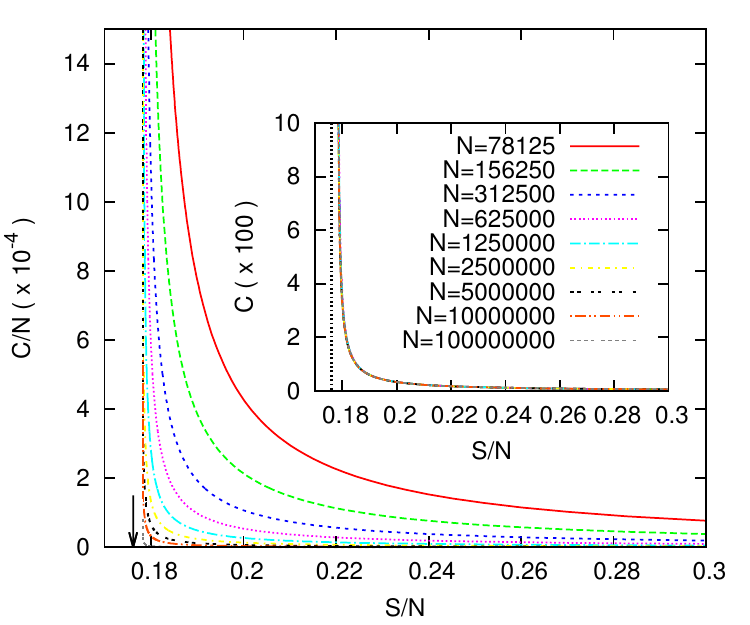}
\end{center}
\caption{Fraction of nodes in the largest component $C/N$ as a function of
  the fraction of nodes removed by the MS algorithm followed by
  the greedy tree breaking for the ER random graph
  of average degree $d=3.5$ and a range of sizes. Inset: Same
  plot for the size of the largest component. The collapse of the
  curves suggests that to reduce the largest component to a given
  constant size $C$, it is sufficient to remove $S = s N$ nodes, where $s=s(C)$ does not depend on $N$.\label{fig:scaling}} 
\end{figure}

In Fig.~\ref{fig:scaling} we zoom in on the results of the second
stage of our algorithm and perform a finite size scaling analysis,
increasing the size of the dismantled graphs up to $N=10^8$. In this way, we identify a threshold for decycling (and thus for dismantling)
by the MS algorithm that converges towards the value
$\theta_{\rm dec}^{MS} \approx 0.1782$, that is close but not equal
to the theoretical prediction of the 1RSB calculation $\theta_{\rm
  dec}^{\rm 1RSB} \approx 0.1753$ (vertical arrow). The inset of
Fig.~\ref{fig:scaling} shows a remarkable scaling that indicates that the size
of the largest component after dismantling by removing a given fraction of
nodes does not depend on the graph size.

Combinatorial optimization problems typically exhibit a very large
degeneracy of their (quasi)-optimal solutions. We performed a detailed
statistical analysis of the quasi-optimal dismantling sets constructed
by our algorithm, exploiting the fact that the MS algorithm finds
different decycling sets for different realizations of the random
tie-breaking noise $\varepsilon_i(t_i)$.

For a given ER random graph of average degree $3.5$ and size $N=78125$ we ran the
algorithm for 1000 different realizations of the tie-breaking
noise $\varepsilon_i(t_i)$ and obtained 1000
different decycling sets, all of which had sizes within 40 nodes of
one another. Randomly chosen pairs among these 1000
decycling sets coincided, on average, on 82\% of their nodes. 
For each node, we computed its frequency of appearance among
the 1000 decycling sets we obtained. We then ordered nodes by this frequency and plotted the
frequency as a function of this ordering in Fig.~\ref{fig:frequency}. 
We see that some nodes appear in almost all found sets, about $60\%$
of nodes does not appear in any, and a large portion of nodes appear only in a
fraction of the decycling sets. We compare the frequencies of nodes
belonging to one typical set found by MS and by the CI heuristics. 

\begin{figure}[htb]
\begin{center}
\includegraphics[width=0.5\columnwidth]{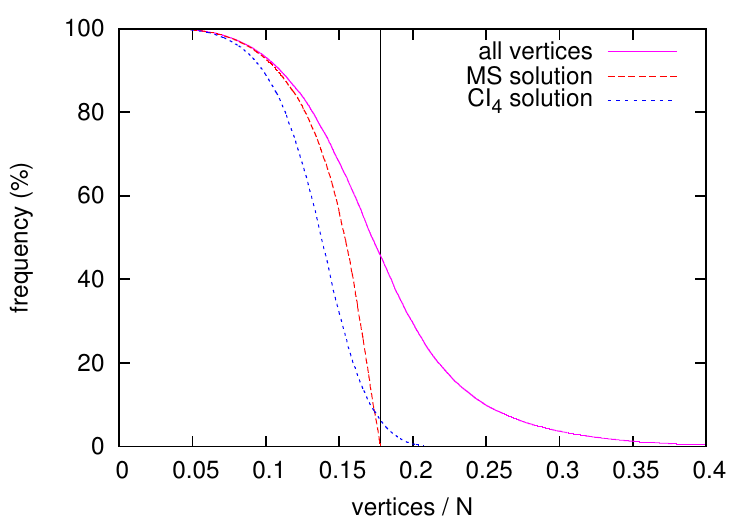}
\end{center}
\caption{Frequencies with which vertices appear in different
  decycling sets, on an ER graph,
  $N=78125, d=3.5$. The y-axis gives the frequency with which a given
  vertex appears in close-to-optimal decycling sets found by the MS
  algorithm. We ordered the vertices by this frequency and depict their
  ordering divided by $N$ on the x-axis. The different
  curves correspond to all vertices, vertices appearing in one
  randomly chosen decycling set found by the MS algorithm, and in one
  found by the CI algorithm. 
 \label{fig:frequency}} 
\end{figure}

An important question to ask about dismantling sets is whether they
can be thought of as a collection of nodes that are in some sense good
{\it spreaders} or whether they are a result of highly correlated
optimization. We use the result of the previous experiment and remove
the nodes that appeared most often, i.e have the highest frequencies in Fig.~\ref{fig:frequency}. If the nature of
dismantling was additive rather than collective then such a choice
should further decrease the size of the discovered dismantling
set. This is not what happens, with this strategy we need to remove
$20.1\%$ of nodes in order
to dismantle the graph, compared to the $17.8\%$ of nodes found
systematically by the MS algorithm. From this we conclude that dismantling is an intrinsically
collective phenomenon and one should always speak of the full {\it set} rather than of
a collection of {\it influential spreaders}.

We also studied the degree histogram of nodes that the MS algorithm
includes in the dismantling sets and saw that, as expected, most
of the high-degree nodes belong to most of the dismantling
sets. Each of the dismantling sets also included some nodes of relatively low degrees;
for instance, for an ER random graph of average degree $d = 6$ and size $5^7$
a typical decycling set found by the MS algorithm has around $460$ (i.e.
around $17\%$ of the decycling set) nodes of degree $4$ or lower. To assess the
importance of low degree nodes for dismantling, we ran the MS
algorithm under the constraint  that only nodes of degree at least $5$ can be removed, we find decycling sets almost as small (only
about 50 nodes, i.e. $0.2\%$ larger) as without this constraint.
From this we conclude that none of the low degree nodes (even those with high CI
centrality) is indispensable for dismantling, going against a
highlight claim of~\cite{morone_influence_2015}. 
 
\begin{figure}[htb]
\begin{center}
\includegraphics[width=0.5\columnwidth]{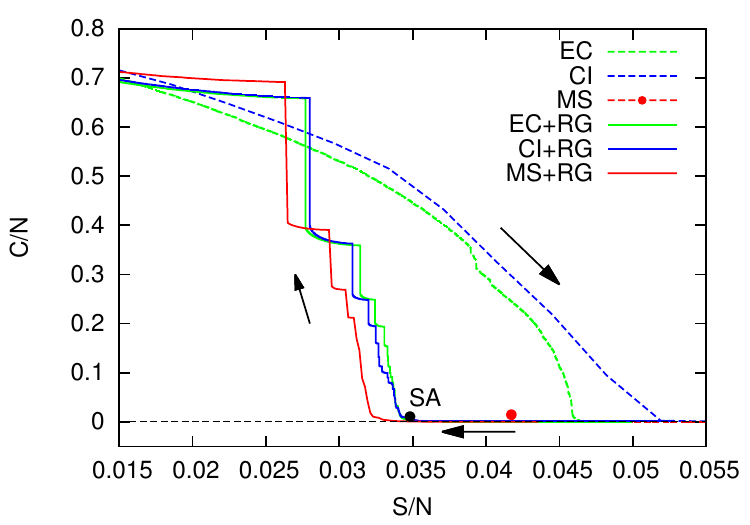}
\end{center}
\caption{Fraction $C/N$ taken by the largest component in the Twitter network achieved
  after removing a fraction $S/N$ of nodes using the Min-Sum algorithm
  (MS) and the adaptive versions of ``Collective'' Influence (CI) and Eigenvector Centrality (EC)
 measures. The red point marks the result obtained by decycling using
 MS (followed by the curve from the optimal tree-breaking
 process). The branches at lower values of $S/N$ are obtained after
 the application of the reverse greedy strategy (RG) from the graph
 obtained when the largest component has $C=100$ nodes. The black point denotes the dismantling fraction
  obtained by Simulated Annealing (SA).
 \label{fig:Twitter}} 
\end{figure}

\subsection{More general graphs}
Up to this point our study of dismantling relies crucially on the
relation to decycling. For light-tailed random graphs these two
problems are essentially asymptotically equivalent. But for arbitrary graphs, that
contain many small cycles, the decycling number can
be much larger than the dismantling one. We argue that, from the algorithmic point of view, decycling still
provides a very good basis for dismantling. For instance, consider a portion of
$N=532 000$ nodes of the Twitter network already analyzed in
\cite{morone_influence_2015}. The decycling solution found by MS
improves considerably the results obtained with the CI and EC heuristics
(see Fig.~\ref{fig:Twitter}). 

In a network that contains many short cycles cycles, decycling removes a large proportion of nodes expressly to destroy these short cycles. Many of these
nodes can be put back without increasing the size of the largest
component. For this reason we introduce a {\it reverse
greedy} (RG) procedure, in which, starting from a dismantled graph with
dismantling set $S$, maximum component size $C$ and a chosen target
value $C'>C$ for the maximum allowed component size, removed nodes are iteratively reinserted. At each step, among all removed nodes, the one which 
ends up in the smallest connected component is chosen for reinsertion (see Appendix~\ref{sec_greedy_reintroduction} for details). The computational cost of this operation is bounded by
$k_{\rm max} C' \log (k_{\rm max} C')$, where $k_{\rm max}$ is the maximal degree of the graph; the update cost is thus typically sublinear in $N$.

 In graphs where decycling is an optimal strategy for dismantling,
 such as the random graphs, a vanishing fraction of nodes can be reinserted by the RG procedure before the size of the largest component
 starts to grow steeply. For real-world networks, the RG procedure reinserts a considerable number of nodes, negligibly altering the size of the largest component. For the Twitter network
 in Fig.~\ref{fig:Twitter}, the improvement obtained by applying the
 RG procedure is impressive, $32\%$ fewer nodes
 for the CI method, and $20\%$ fewer nodes for the MS algorithm, which ends up being the
 best solution we found, removing only $3.4\%$ of nodes in order to
 dismantle into components smaller that $C=1000$ nodes. RG makes possible to reach, and even improve, the best result obtained with Simulated Annealing (SA) that solves the dismantling problem directly and is not affected by the presence of short loops (see Appendix~\ref{sec_otheralgo} for details on SA). 
 Qualitatively similar results are achieved on
 other real networks, e.g. on the Youtube network with $1.13$ million nodes \cite{snapnets} the best
 dismantling set we found with MS+RG included $4.0\%$ of nodes, this is
 a $22\%$ improvement with respect to the CI heuristics. 

The reverse-greedy procedure is introduced as a heuristic that
provides a considerable improvement for the examples we
treated. The theoretical results of this paper are valid only for
classes of graphs that do not contain many small cycles and hence our
theory does not provide a principled derivation nor analysis of the RG
procedure. This is an interesting open direction for future work.   
More detailed study (both theoretical and algorithmic) of dismantling of
networks for which decycling is not a reasonable starting point is an important direction of
future work. 






We provide in the appendices further technical details
and additional results in support of the main text. They are organized as
follows. In Appendix~\ref{sec_NPcomplete} we prove the
NP-completeness of the dismantling decision problem. 
In Appendix.~\ref{sec_analytic} we extend the analytic results of the main
text, presenting the details of the cavity method computation of the
decycling number of random graphs (\ref{sec_cavity}), a lower bound
on the decycling number
valid for all graphs (\ref{sec_lowerbound}), and an expansion of the
decycling number for Erd\H os-R\'enyi random graphs close to their
percolation threshold and for large average degrees
(\ref{sec_dec_expansions}).
Appendix~\ref{sec_algorithms} is then devoted to several algorithmic
aspects: in~\ref{sec_MS}, \ref{sec_tree_breaking} and 
\ref{sec_greedy_reintroduction} we detail the three stages of our main 
algorithm (derivation of the Min-Sum equations, tree dismantling and greedy 
reintroduction of cycles respectively), 
while in \ref{sec_otheralgo} we give further
details on the other dismantling algorithms we have studied.
Finally in Appendix~\ref{sec_othergraphs} we provide further results on other
real-world and artificial scale-free networks.

\appendix

\section{Proof of NP-Completeness of the dismantling problem}
\label{sec_NPcomplete}

For our proof, we will employ the decisional (minimum) Vertex Cover problem, which is NP-Complete, and is defined as follows. Remember that a vertex cover is a subset of vertices $W\subset V$ such that for each $(i,j)$ in $E$, $i\in W$ or $j\in W$. 

\smallskip

\noindent \textsc{Vertex Cover}: Given a graph $G=(V,E)$ and $F\in\mathbb{N}$,
does a vertex cover $W\subset V$ with $|W|\leq F$ exist?

\smallskip

\noindent The \textsc{Vertex Cover} problem is NP-Complete.

\smallskip

\noindent \textsc{$C\left(N\right)$--Dismantling}: Given a graph
$G'=(V',E')$ and $F\in\mathbb{N}$, does a $C\left(\left|V'\right|\right)$--dismantling
set $S$ with $|S|\leq F$ of $G'$ exist?
\begin{thm}
Assume $C:\mathbb{N\to\mathbb{N}}$ to be a non-decreasing (polynomially
computable) function with $C\left(N\right)<N^{a}$ for $N\geq N_{0}$
with $0\leq a<1$. Then the $C\left(N\right)$--\textsc{Dismantling} problem
is NP-Complete.\end{thm}
\begin{proof}
\textsc{$C\left(N\right)$--Dismantling} belongs clearly to NP. If
$C\equiv1$, one can see that \textsc{$C$--Dismantling} is identical
to \textsc{Vertex Cover} and is thus is NP--Complete. Otherwise, take
$N_{1}$ such that $C\left(N_{1}\right)\geq2$ and consider $N\geq\max\left\{ N_{0},N_{1}\right\} .$
Define 
\begin{equation}
C'=C'\left(N\right)=\min\left\{ k\in\mathbb{N}:\frac{C(kN)}{k}<2\right\} .\label{eq:min}
\end{equation}

\begin{rem*}
For constant $C\left(N\right)\equiv C\geq2$, then $C'\left(N\right)\equiv\frac{C+1}{2}$
if $C$ is odd, and $C'\left(N\right)\equiv\frac{C}{2}+1$ if $C$
is even. 
\end{rem*}
Note that $C'\geq2$ and $C'$ is polynomial in $N$: the value $k=\left\lceil \frac{1}{2}N^{a}\right\rceil ^{\frac{1}{1-a}}$
belongs to the set in the RHS of \eqref{eq:min}, as $k\geq\left(\frac{1}{2}N\right)^{\frac{a}{1-a}}$
so $k^{a-1}\leq\left(\frac{1}{2}N\right)^{-a}$ and then $\frac{C\left(kN\right)}{k}<\frac{\left(kN\right)^{a}}{k}=k^{a-1}N^{a}\leq2$;
so $C'\leq k=\left\lceil \frac{1}{2}N^{a}\right\rceil ^{\frac{1}{1-a}}$.
We will prove that 

\begin{eqnarray}
C'\leq & C\left(C'N\right) & <2C'.\label{eq:both}
\end{eqnarray}
The second inequality in \eqref{eq:both} follows from \eqref{eq:min}.
For the first inequality, 
\begin{eqnarray}
C\left(C'N\right) & \geq & C\left(\left(C'-1\right)N\right)\label{eq:first}\\
 & \geq & 2\left(C'-1\right)\label{eq:second}\\
 & \geq & C',\label{eq:third}
\end{eqnarray}
where \eqref{eq:first} follows from the fact that $C$ is non-decreasing,
\eqref{eq:second} follows from the minimality of $C'$ in its definition
\eqref{eq:min} and \eqref{eq:third} from the fact that $C'\geq2$. 

Now, take a graph $G$ with $\left|V\right|=N$, and construct $G'$
by adding $C'-1$ leaves to any vertex of $G$. Precisely, let $G'=(V',E')$
with $V'=V\cup\left\{ \ell_{v,i}:v\in V,i\in\{1,\dots,C'-1\right\} \}$
and $E'=E\cup\left\{ (v,\ell_{v,i}):v\in V,i\in\{1,\dots,C'-1\}\right\} $
where we assume the two unions to be disjoint. The number of vertices
of $G'$ is $\left|V'\right|=N'=NC'$ (which is polynomial in $N$).
The construction of $G'$ is clearly polynomial. 

Take any vertex cover $W$ of $G$. Then $W$ is a $C'$--dismantling
of $G'$: thanks to the vertex cover property, each $v\in V\setminus W$
can only be connected to the $C'-1$ extra leaves $\ell_{v,i}$. As
$C'\leq C\left(N'\right)$, then $W$ is also a $C\left(N'\right)$--dismantling
of $G'$.

Conversely, take any $C\left(N'\right)$--dismantling set $S$ of
$G'$. Define $\pi:V'\to V$ by $\pi\left(v\right)=v$ for $v\in V$
and $\pi\left(\ell_{v,i}\right)=v$ for $v\in V,1\leq i\leq C'-1$.
Consider the set $W=\pi\left(S\right)$. In short, $W$ is constructed
from $S$ by replacing all occurrences $\ell_{v,i}$ by $v$. Then
clearly $\left|W\right|=\left|\pi\left(S\right)\right|\leq\left|S\right|$
and $W$ is still a $C\left(N'\right)$--dismantling of $G'$: replacing
$\ell{}_{v,i}$ by $v$ introduces a new component $\left\{ \ell_{v,i}\right\} $
of size 1 but can only reduce the size of the other components. Moreover,
$W$ is also a vertex cover of $G$: suppose on the contrary that
it is not, and take an edge $(i,j)\in E$ such that $i,j\notin W$.
Then both vertices belong to a connected component of $G'\setminus W$
of size $2C'>C\left(N'\right)$, which contradicts the fact that $W$
was a $C\left(N'\right)$--dismantling of $G'$. Thus, $W$ must be
a vertex cover of $G$ with size no greater than $\left|S\right|$
and that proves the result. \end{proof}
\begin{cor}
For $C\left(N\right)=const$, $C\left(N\right)=\log N$, and $C\left(N\right)=N^{a}$
with $0\leq a<1$, $C\left(N\right)$--\textsc{Dismantling} is NP-Complete.\end{cor}
\begin{rem*}
$\left(N-k\right)$--\textsc{Dismantling }is polynomial for any constant
$k$. \end{rem*}


\section{Analytic results}
\label{sec_analytic}
\subsection{Details on the cavity equations for the decycling number
  of random graphs}
\label{sec_cavity}
We give here some more details on the cavity method computation of the decycling
number of sparse random graphs, in particular on the derivation and solution of the BP equations. A full
derivation in a more general context can be found in~\cite{Guilhem14}.

When computing the typical free-energy of a large random
graph with degree distribution $q$ one has to determine the probability law $P(\eta)$ of the messages $\eta$, 
which is the solution of an integral equation of the form:
\beq
P(\eta) = \sum_{k=0}^\infty \tq_k \int \dd P(\eta^{(1)}) \dots \dd
P(\eta^{(k)}) \ \delta(\eta - f_k(\eta^{(1)},\dots,\eta^{(k)})) \ ,
\label{eq:popu}
\eeq
where $\tq_k=(k+1) q_{k+1}/\sum_{k} k q_k$ is the size-biased
distribution associated to $q$ (i.e. the probability of finding a
vertex of degree $k+1$ when choosing an edge uniformly at random), and
$f_k$ the function encoding the local BP equation, eq.~(\ref{eq:BPeq}) in the main
text, between messages
around a vertex of degree $k+1$. This type of equation can be
efficiently solved numerically via a population dynamics procedure, in
which $P$ is approximated by a large sample of representative values
of $\eta$, updated according to \eqref{eq:popu} until convergence to a
fixed point. The free-energy density of the model can then be computed
as the average with respect to $P$ of suitable functions of the
messages. In the present model these messages are real vectors of a
dimension which grows linearly with the parameter $T$ introduced above
as a cutoff on the allowed times in the leaf-removal dynamics.

In the Replica Symmetric version of the cavity method a message (or
field) $\eta$ of \eqref{eq:popu} corresponds to a $2T$ dimensional vector
of components denoted $(a_0,a_1,\dots,a_T,b_{T-1},\dots,b_1)$. The
function $f_k$ which gives $\eta$ as a function of
$\eta^{(1)},\dots,\eta^{(k)}$ reads explicitly:
\bea
e^{-\mu b_t} &=& 1 + e^{-\mu+\mu\overset{k}{\underset{i=1}{\sum}} (a_0^{(i)} -
  b_{t-1}^{(i)})} \ ,\\
e^{-\mu a_t} - e^{-\mu a_{t+1}} &=&  e^{-\mu+\mu\overset{k}{\underset{i=1}{\sum}} (a_0^{(i)} -
  b_t^{(i)})} 
\sum_{i=1}^k \left( e^{\mu (b_t^{(i)} - a_{t+1}^{(i)}) } -
e^{\mu (b_t^{(i)} - a_{t+2}^{(i)}) }  \right) \ ,
\nonumber
\eea
with the conventions used to have more compact expressions: $b_0=0$,
$b_T=a_T$, $a_{T+1}=b_{T-1}$. Once the self-consistent equation on
$P(\eta)$ is solved the thermodynamic quantities are obtained as
follows. The limit of $(\ln Z)/N$ reads
\beq
\phi = \mu + \Exp[\ln z_{\rm site}(\eta^{(1)},\dots,\eta^{(k)})] -
\frac{d}{2} \, \Exp[\ln z_{\rm edge}(\eta^{(1)},\eta^{(2)})] \ , \nonumber
\eeq
where $\Exp[\cdot]$ denotes the average over the i.i.d. copies
$\eta^{(i)}$ drawn from $P(\eta)$ and over the integer $k$ drawn from the
degree distribution $q$, and $d$ is the mean of $q$. The two functions
$z_{\rm site}$ and $z_{\rm edge}$ arise from the local contributions
to the Bethe free-energy of sites and edges respectively, and read
\bea
z_{\rm site} &=& 1+ e^{-\mu+\mu\overset{k}{\underset{i=1}{\sum}}
  a_0^{(i)}} \left[
e^{-\mu\overset{k}{\underset{i=1}{\sum}} b_{T-1}^{(i)}} + 
\sum_{t=1}^T
e^{-\mu\overset{k}{\underset{i=1}{\sum}} b_{t-1}^{(i)}} 
\sum_{i=1}^k \left( e^{\mu (b_{t-1}^{(i)} - a_t^{(i)}) } -
e^{\mu (b_{t-1}^{(i)} - a_{t+1}^{(i)}) }  \right)
\right]  \ , 
\\
z_{\rm edge}  &=& e^{\mu (a_0^{(1)} + a_0^{(2)} )} \left[ 
e^{-\mu (b_T^{(1)} +b_T^{(2)}  )} + 
\sum_{t=0}^{T-1} \left\{
(e^{-\mu a_t^{(1)}} - e^{-\mu a_{t+1}^{(1)}} ) e^{-\mu b_t^{(2)}} +
(e^{-\mu a_t^{(2)}} - e^{-\mu a_{t+1}^{(2)}} ) e^{-\mu b_t^{(1)}} 
\right\}
\right] 
\ .
\eea
The fraction $\theta$ of vertices included in the decycling sets
selected by the conjugated chemical potential $\mu$ and the entropy $s$
(the Legendre transform of $\phi$) then read:
\beq
\theta = \Exp \left[ \frac{1}{z_{\rm site}(\eta^{(1)},\dots,\eta^{(k)})}
\right] \ , \ \ \ s = \phi - \mu \, \theta \ .
\eeq
Varying the parameter $\mu$ one can compute in this way the entropy
$s(\theta)$ counting the exponential number of decycling sets
containing a fraction $\theta$ of vertices. The RS estimate of the
decycling number $\theta_{\rm dec}$ is then obtained as the point
where $s$ vanishes.

This estimate is, however, only a lower bound to the true value of
$\theta_{\rm dec}$ because of the effects of the replica symmetry
breaking. A more precise estimate is obtained by using the (energetic)
cavity method at the first level of replica symmetry breaking (1RSB),
in which the parameter $\mu$ is replaced by the Parisi breaking
parameter $y$; the message $\eta$ of \eqref{eq:popu} is then a vector
$(p_0,\dots,p_{T-1},r_T,\dots,r_0)$ constrained by the normalization
$p_0+\dots+p_{T-1}+r_T+\dots+r_0=1$ (hence the number of independent
parameters is again $2T$). These are updated according to 
\bea
p_t &=& \frac{1}{Z}e^y \, \tp_t  \ , \\
r_t &=& \frac{1}{Z}e^y \, \tr_t \ \ \text{for} \ \ t \ge 1 \ , \\
r_0 &=& \frac{1}{Z} \left( 1 - \sum_{t=0}^{T-1} \tp_t - \sum_{t=1}^T
  \tr_t \right) \ , \\ 
Z &=& 1 + (e^y -1) \left(  \sum_{t=0}^{T-1} \tp_t + \sum_{t=1}^T \tr_t
\right) \ , \\
\tp_t &=& \sum_{i=1}^k p_{t+1}^{(i)} \prod_{j \neq i}
\left(\sum_{t'=0}^t r_{t'}^{(j)} \right)\ , \\
\tr_t &=& \prod_{i=1}^k \left(\sum_{t'=0}^{t-1} r_{t'}^{(i)} \right) 
- \prod_{i=1}^k \left(\sum_{t'=0}^{t-2} r_{t'}^{(i)} \right) \ ,
\eea
with $p_T=r_T$ by convention.
One computes then a thermodynamic potential $\Phi(y)$ with a formula
similar to the one yielding $\phi(\mu)$ at the RS level, namely
\beq
\Phi = - y + \Exp[\ln \Z_{\rm site}(\eta^{(1)},\dots,\eta^{(k)})] -
\frac{d}{2} \, \Exp[\ln \Z_{\rm edge}(\eta^{(1)},\eta^{(2)})] \ , \nonumber
\eeq
with
\bea
\Z_{\rm site} &=& 1+ (e^y -1) 
\left[
\prod_{i=1}^k \left(\sum_{t=0}^{T-1} r_t^{(i)} \right)
+ 
 \sum_{i=1}^k  \sum_{t=1}^T p_t^{(i)} \prod_{j\neq i} \left(\sum_{t=0}^{t-1} r_t^{(j)} \right)
\right] \ ,
\\
\Z_{\rm edge}  &=& e^{-y} + (1-e^{-y}) \left[ 
\left(\sum_{t=0}^T r_t^{(1)} \right) \left(\sum_{t=0}^T r_t^{(2)}
\right)
+ 
\sum_{t=0}^{T-1} 
\left\{ p_t^{(1)} \left(\sum_{t'=0}^t r_{t'}^{(2)} \right) +
p_t^{(2)} \left(\sum_{t'=0}^t r_{t'}^{(1)} \right) \right\}
\right] \ . 
\eea
The energetic complexity function (the equivalent of the entropy at
the 1RSB level) is then obtained by an inverse Legendre transform with
respect to $\Phi$, namely
\beq
\Sigma = \Phi + y \, \theta \ , \ \ \ \ \theta = 1- \Exp\left[\frac{\Z'_{\rm site}}{\Z_{\rm site}}\right] +
\frac{d}{2} \, \Exp\left[\frac{\Z'_{\rm edge}}{\Z_{\rm edge}}\right] \ , \nonumber
\eeq
where the prime denotes the derivative with respect to the explicit
dependence in $y$ of the expressions of $\Z_{\rm site}$ and $\Z_{\rm
  edge}$ given above. The 1RSB estimate of the decycling number is
then obtained from the criterion of cancellation of the complexity
$\Sigma$. Both the replica symmetric and 1RSB results for a range of
values of $T$ are reported in Table~\ref{tab:RS1RSB}. Extrapolating
the 1RSB estimate of the decycling number in the limit $T \to \infty$ leads to the values reported
in Table~\ref{tab:ER1RSB} in the main text.

The replica symmetric and 1RSB computations yield improving lower bounds on the decycling number of light-tailed random graphs, in the sense that $\theta^{\rm RS}(q) \le \theta^{\rm 1RSB}(q) \le \theta_{\rm dec}(q)$. For some degree distributions $q$ these inequalities become equalities (see~\cite{Guilhem14} for details on random regular graphs), for others the 1RSB estimate is strictly tighter than the RS one. It is probable that for some choices of $q$ the 1RSB estimate is not equal to the decycling number, whose exact determination would require the use of the so-called full RSB computation. The latter is not tractable numerically for models of sparse random graphs, we expect in any case the quantitative difference between the 1RSB and full RSB results to be rather small.

\begin{table}
\begin{tabular}{ccc}
$T$ & $\theta^{\rm RS}(T)$ & $\theta^{\rm 1RSB}(T)$ \\
\hline
\hline
3 &   0.22714 &  0.22797 \\
\hline
5 &   0.20042 &  0.20077\\
\hline
9 &   0.18507 &  0.18515 \\
\hline
13 & 0.18046 &  0.18051 \\
\hline
19 & 0.17795  & 0.17797 \\
\hline
30 & 0.17638  &  0.17638 \\
\hline
40 & 0.17590  &  0.17590 \\
\hline
50 & 0.17569  & 0.17569 \\
\hline
\end{tabular}
\caption{The $T$-dependence of the RS and 1RSB cavity predictions for the decycling number of
  Erd\"{o}s-R\'enyi random graphs of average degree $d=3.5$. The decycling
numbers reported in table~1 in the main text are obtained by
interpolation of the 1RSB
results to $T\to\infty$.}
\label{tab:RS1RSB}
\end{table}

\subsection{A simple lower bound}
\label{sec_lowerbound}

We present here a lower bound on the decycling number $\theta_{\rm
  dec}(G)$ valid for any graph $G$ (a similar reasoning can be found 
in~\cite{decycling_Beineke}), generalizing the bound $\theta_{\rm
  dec}(G)\ge \frac{d-2}{2(d-1)}$ for $d$-regular graphs.

We denote $k_i$ the degree of vertex $i$ and $M$ the number of edges. With
\beq
\la k \ra = \frac{1}{N} \sum_{i \in V} k_i
\eeq
the empirical average degree, one has $M = N \frac{\la k \ra}{2}$.

Consider now a subset $S$ of the vertices, and its complement $S^c = V \setminus S$. One can divide the edges in three categories, with $M=M_1 + M_2 + M_3$, where $M_1$ is the number of edges between two vertices of $S$, $M_2$ counting the edges between $S$ and $S^c$, and $M_3$ the edges inside $S^c$. One has
\beq
\sum_{i \in S} k_i = 2 M_1 + M_2 \ , \qquad \sum_{i \in S^c} k_i = 2 M_3 + M_2 \ , 
\eeq
and in particular
\beq
M_1 + M_2 \leq \sum_{i \in S} k_i \ .
\eeq
Suppose now that $S$ is a decycling set of the graph, in such a way that $S^c$ induces a forest. Hence one has
\beq
M_3 \leq |S^c| - 1 = N - |S| - 1 \ .
\eeq
Summing these two inequalities, and expressing $M$ in terms of the average degree, yields
\beq
\frac{1}{N} \sum_{i \in S} (k_i - 1) \ge \frac{\la k \ra}{2} - 1 + \frac{1}{N} \ .
\eeq
This inequality constrains the possible decycling sets. To obtain a
simpler lower bound on the size of the decycling sets, consider a
permutation $\sigma$ from $\{1,\dots,N\}$ to $V$ that orders the
vertices according to their degrees: $k_{\sigma(1)} \ge k_{\sigma(2)}
\ge \dots$. The inequality above can then be continued to get
\beq
\frac{1}{N} \sum_{i=1}^{|S|} (k_{\sigma(i)} - 1) \ge \frac{\la k \ra}{2} - 1 + \frac{1}{N} \ .
\label{eq_ineq}
\eeq
Let us call the left hand side of this inequality $l(\theta=|S|/N,G)$,
which is an increasing function of $\theta$, and define $\theta_{\rm lb}(G)$ as the smallest value of $\theta$ such that the inequality is fulfilled. Then the decycling number $\theta_{\rm dec}(G)$ of this graph is certainly lower-bounded by $\theta_{\rm lb}(G)$.

The shape of $l(\theta)$ can be described in terms of the empirical degree distribution of the graph,
\beq
\widehat{q}_k = \frac{1}{N} \sum_{i \in V} \delta_{k,k_i} \ .
\eeq
As the graph is finite so is its maximal degree, let us call it
$K$. One realizes easily that $l(\theta,G)$ is a piecewise linear
continuous increasing function, starting from 0 in $\theta=0$, linearly increasing on $\theta \in [0,\widehat{q}_K]$ with slope $K-1$, then again with a constant slope $K-2$ on the interval $\theta \in [\widehat{q}_K,\widehat{q}_K +\widehat{q}_{K-1}]$, and so on and so forth. It is thus more convenient to introduce two integrated quantities:
\beq
\widehat{Q}_k = \sum_{k'=k}^\infty \widehat{q}_{k'} \ , \qquad 
\widehat{T}_k = \sum_{k'=k}^\infty \widehat{q}_{k'} (k'-1) \ ,
\eeq
the summations being cut off at $K$ in this finite graph case. Indeed
for all $k$ one has $l(\widehat{Q}_k,G)=\widehat{T}_k$,
and the function $l(\theta,G)$ is the linear interpolation
between this discrete set of points. One can thus determine its intersection with the right hand side of (\ref{eq_ineq}) to compute the lower bound $\theta_{\rm lb}(G)$.

In the case of random graphs drawn with a degree distribution $q$ the
typical decycling number $\theta_{\rm dec}(q)$ can be lower-bounded as
above by replacing the empirical distribution $\widehat{q}$ by
$q$: $\theta_{\rm lb}(q) = l^{-1}\left(\frac{\la k \ra}{2} - 1 \right)$, where $\la k \ra$ is now averaged with respect to $q_k$, and $l$ is defined by replacing $\widehat{Q}_k$ and $\widehat{T}_k$ by their counterparts
\beq
Q_k = \sum_{k'=k}^\infty q_{k'} \ , \qquad 
T_k = \sum_{k'=k}^\infty q_{k'} (k'-1) \ .
\eeq

The numerical evaluation of this lower bound for a Poissonian random
graph of average degree $d=3.5$ yields $\theta_{\rm lb} = 0.141084$, not
that far from the 1RSB prediction $\theta_{\rm dec}=0.175$. The
lower bound matches the asymptotic expansion presented below when $d\to
1$ (i.e. close to the percolation threshold), while it reaches the
limit $1/2$ when $d$ diverges (the large $d$ limit of $\theta_{\rm dec}$ being 1).

\subsection{Decycling close to the percolation threshold and for large
degrees}
\label{sec_dec_expansions}
In addition to the numerical results obtained by the cavity method let us state analytical
asymptotic expansions for the decycling number of Poissonian random
graphs with average degree close to the percolation threshold
($d=1+\epsilon$) or very large ($d\to\infty$). Close to the
percolation a random graph is essentially made
of a 3-regular kernel of vertices joined by paths of degree 2 nodes;
decycling the kernel is sufficient to decycle the whole graph, and the
decycling number of a random 3-regular is
known~\cite{decycling_Wormald}, which yields
\beq
\theta_{\rm dec}(d=1+\epsilon) = \frac{1}{3} \epsilon^3 +
O(\epsilon^4) \ .
\eeq
On the other hand when $d$ is very large the Poissonian random graph
behaves like a regular graph (the degree distribution being
concentrated around its average), an asymptotic expansion in this case
was obtained in~\cite{Guilhem14} (in agreement with the rigorous
bounds of~\cite{lb_klequal_rig}), hence
\beq
\theta_{\rm dec}(d)= 1- \frac{2 \ln d}{d} - \frac{2}{d} + O\left(\frac{1}{d \ln d} \right) \ .
\eeq

\section{Algorithms}
\label{sec_algorithms}
\subsection{The Min-Sum algorithm and its implementation}
\label{sec_MS}

In this section we derive the Min-Sum (MS) algorithm introduced in eqs.~(\ref{eq:field1b-main}-\ref{eq:ms_h00}) of
the main text, that aims at finding decycling sets of the smallest
possible size. As explained in the main text this amounts to find the
unique minimum of the cost function $\sum_i \psi_i(t_i)$, with 
$\psi_{i}\left(t_{i}\right)=\mathbb{I}\left[t_{i}=0\right] +
\varepsilon_i(t_i)$, over the feasible configurations of the
activation times $\{t_i\}\in\{0,\dots,T\}^V$. These variables have to
fulfill the constraint that for all
vertices $i$ either $t_i=0$ (when $i$ belongs to the decycling set) or
it is determined by the adjacent variables according to
$t_i = 1+ {\rm max}_2(\{t_j\}_{j \in \partial i} )$. We recall that
the $\varepsilon_i(t_i)$ are infinitesimally small random variables
that are introduced to ensure the uniqueness of the minimum of the
cost function, and its closeness to one of the minima of the original
cost function. In practice we took $\epsilon_i$ to be uniformly random
between 0 and $10^{-7}$. 

It turns out to be easier to study a slight modification of this
optimization problem, with a relaxed constraint
\beq
t_i \ge 1+ {\rm max}_2(\{t_j\}_{j \in \partial i} ) \qquad \text{if} \
\ 0 < t_i \le T 
\label{eq_constraint1}
\eeq
that corresponds to a lazy version of the leaf removal algorithm, in
which a node can be removed once it became a leaf, but it is not
necessarily removed as soon as it could be. Thanks to the monotonicity
of the leaf removal procedure its final outcome, the 2-core of the
graph, is independent of the order in which the leaves are removed, and
of the parallel or sequential character of these updates. Hence the
optimization problem with the strict or relaxed constraints are
completely equivalent if $T \ge N$, the maximal number of possible
steps of the leaf removal. For smaller values of $T$ this equivalence
is not ensured, but the optimum with the relaxed constraints still
provides a valid decycling set. It will be useful in the following to
use the following logical equivalent of (\ref{eq_constraint1}),
\begin{equation}
  \sum_{j\in\partial i}\mathbb{I}\left[t_{j}\geq t_{i}\right]\leq
  1\qquad \text{if} \ 0<t_i\leq T \ .
\label{eq:relaxed}
\end{equation}
Our goal now is to compute the field $h_i(t_i)$ defined as the
minimum of the cost function over feasible configurations with a given
value of $t_i$, as indeed the unique minimum can be deduced from the
fields through $t_i^*={\rm argmin} \, h_i(t_i)$, and then the corresponding
decycling set is identified with the vertices $i$ where $t_i^*=0$. To
justify the MS heuristics for the approximate computation of the fields
$h_i$ on any graph it is simpler to consider first a tree graph, on
which the MS approach is exact for any local cost function. 
The function under consideration
here is the sum of local terms $\psi_i(t_i)$ on each vertex,
$h_i(t_i)$ can thus be decomposed as a sum of its own contribution $\psi_i$ and
of the contributions of the vertices in each of the subtrees rooted at
one of its neighbor $j \in \partial i$. Taking into account the constraint
(\ref{eq:relaxed}) for the positive times, denoted  $\mathcal C$ in the following, it yields
\begin{eqnarray}
  h_{i}\left(t_{i}\right) & = & \psi_{i}\left(t_{i}\right) + \min_{\{t_j\}_{j\in\partial
      i}:\mathcal{C}}\sum_{j\in\partial
    i}h_{ji}\left(t_{j},t_{i}\right)\qquad
  \text{for} \ 0 <t_i \le T \ ,
\label{eq:field1}\\
h_{i}\left(0\right) & = & \psi_{i}\left(0\right) + \sum_{j\in\partial
  i}\min_{t_{j}}h_{ji}\left(t_{j},0\right)\ , 
\label{eq:field2}
\end{eqnarray}
where $h_{ji}(t_j,t_i)$ are messages defined on each directed edge 
$j \to i$ of the graph, that give the minimum cost of the variables
in the subtree rooted at $j$ and excluding $i$, over the feasible
configurations with prescribed values of $t_j$ and $t_i$. Thanks to
the recursive structure of a tree these messages obey themselves
similar equations,
\begin{eqnarray}
  h_{ij}\left(t_{i},t_{j}\right) & = & \psi_{i}\left(t_{i}\right) + \min_{\{t_k\}_{k\in\partial i\setminus j}:\mathcal{C}}\sum_{k\in\partial i\setminus j}h_{ki}\left(t_{k},t_{i}\right) \qquad
  \text{for} \ 0 <t_i \le T \ , \label{eq:MS1}\\
h_{ij}\left(0,t_{j}\right) & = & \psi_{i}\left(0\right) +
\sum_{k\in\partial i\setminus j}\min_{t_{k}}h_{ki}\left(t_{k},0\right)
\ .
\end{eqnarray}
These Min-Sum equations involve $O(T^2)$ quantities for each edge of
the graph  because of the two time indices of the messages
$h_{ij}$. Fortunately this quadratic dependence on $T$ can be reduced
to a linear one by some further simplifications that we now explain.

As can be readily seen from \eqref{eq:MS1} and \eqref{eq:relaxed}, the dependence of $h_{ij}$ on $t_{j}$ 
is only through $\mathbb{I}\left[t_{j}<t_{i}\right]$. We will thus define 
\begin{equation}
h_{ij}\left(t_{i},t_{j}\right)=\begin{cases}
h_{ij}^{1}\left(t_{i}\right) & \mbox{ if }t_{j}<t_{i} \ , \\
h_{ij}^{0}\left(t_{i}\right) & \mbox{ if }t_{j}\geq t_{i} \ ,
\end{cases} 
\end{equation}
which gives a parametrization of each message with $O(T)$ real
numbers.

Calling $\mathcal{T}^s_{ij}=\{\left\{ t_{k}\right\} _{k\in\partial i\setminus j}:\underset{k\in\partial i\setminus j}{\sum}\mathbb{I}\left[t_{k}\geq t_{i}\right]= s\}$, Eq.~\eqref{eq:MS1} can be rewritten as follows for $0<t_i\leq T$:
\begin{eqnarray}
  h_{ij}^{0}\left(t_{i}\right) & = &
  \psi_{i}\left(t_{i}\right)+\min_{\mathcal{T}^0_{ij}}\sum_{k\in\partial
    i\setminus j}h_{ki}\left(t_{k},t_{i}\right)\nonumber \\
 & = & \psi_{i}\left(t_{i}\right)+\sum_{k\in\partial i\setminus
   j}\min_{t_{k}<t_{i}}h_{ki}\left(t_{k},t_{i}\right) \nonumber \\
 & = & \psi_{i}\left(t_{i}\right)+\sum_{k\in\partial i\setminus
   j}\min_{t_{k}<t_{i}}h_{ki}^{0}\left(t_{k}\right) \ , \label{eq:h0}
 \end{eqnarray}
as indeed $t_j \ge t_i$ in the definition of $h_{ij}^{0}$ all the other
removal times $t_k$ for $k\in \partial i \setminus j$ have to be
strictly smaller than $t_i$ for the condition \eqref{eq:relaxed} to be fulfilled.
On the other hand in the situation described by $h_{ij}^{1}$ at most
one of the removal times $t_k$ for $k\in \partial i \setminus j$ can
be greater or equal than $t_i$, hence for $0<t_i\leq T$:
\begin{eqnarray}
h_{ij}^{1}\left(t_{i}\right) & = &
\psi_{i}\left(t_{i}\right)+\min_{\mathcal{T}^0_{ij}\cup\mathcal{T}^1_{ij}}\sum_{k\in\partial
  i\setminus j}h_{ki}\left(t_{k},t_{i}\right) \nonumber \\
 & = & \psi_{i}\left(t_{i}\right)+\min\bigg\{\sum_{k\in\partial i\setminus j}\min_{t_{k}<t_{i}}h_{ki}^{0}\left(t_{k}\right),
\min_{\mathcal{T}^1_{ij}}\sum_{k\in\partial i\setminus
  j}h_{ki}\left(t_{k},t_{i}\right)\bigg\} \nonumber \\
 & = & \psi_{i}\left(t_{i}\right)+\min\bigg\{\sum_{k\in\partial i\setminus j}\min_{t_{k}<t_{i}}h_{ki}^{0}\left(t_{k}\right),
\min_{k\in\partial i\setminus j}\bigg[\min\left\{ h_{ki}^{0}\left(t_{i}\right),\min_{t_k>t_{i}}h_{ki}^{1}\left(t_k\right)\right\}
+\sum_{\ell\in\partial i\setminus j,k}\min_{t_{\ell}<t_{i}}h_{\ell
  i}^{0}\left(t_{\ell}\right)\bigg]\bigg\} \nonumber \\
 & = & \psi_{i}\left(t_{i}\right)+\sum_{k\in\partial i\setminus j}\min_{t_{k}<t_{i}}h_{ki}^{0}\left(t_{k}\right)+
\min\bigg\{0,\min_{k\in\partial i\setminus j}\bigg[\min\bigg\{ h_{ki}^{0}\left(t_{i}\right),
\min_{t_{k}>t_{i}}h_{ki}^{1}\left(t_{k}\right)\bigg\}-\min_{t_{k}<t_{i}}h_{ki}^{0}\left(t_{k}\right)\bigg]\bigg\}
\ .
\label{eq:h1}
\end{eqnarray}

The equations (\ref{eq:field1b-main}-\ref{eq:ms_h00}) of the main text can now be readily
obtained by defining the following quantities:
\begin{eqnarray}
L_{ki}\left(t_{i}\right) & = &
\min_{t_{k}<t_{i}}h_{ki}^{0}\left(t_{k}\right) \ , \label{eq:ms_L1}\\
R_{ki}\left(t_{i}\right) & = & \min\left\{
  h_{ki}^{0}\left(t_{i}\right),\min_{t_{k}>t_{i}}h_{ki}^{1}\left(t_{k}\right)\right\}
\ , \label{eq:ms_R1} \\
M_{ij}\left(t_{i}\right) & = & \min\left\{0,\min_{k\in\partial
    i\setminus j}
  \{R_{ki}(t_i)-L_{ki}(t_i)\}\right\}  \ , \label{eq:ms_M1}\\
\end{eqnarray}
in terms of which the equations (\ref{eq:h0},\ref{eq:h1}) can be
rewritten
\begin{eqnarray}
h_{ij}^{0}\left(t_{i}\right) & = & \psi_{i}\left(t_{i}\right)+ \sum_{k\in\partial i\setminus j} L_{ki}(t_i)  \qquad
  \text{for} \ 0 <t_i \le T \ , \label{eq:ms_h0t1}\\
h_{ij}^{1}\left(t_{i}\right) & = & \psi_{i}\left(t_{i}\right)+ \sum_{k\in\partial i\setminus j} L_{ki}(t_i)+M_{ij}\left(t_{i}\right) \label{eq:ms_h1t1}\qquad
  \text{for} \ 0 <t_i \le T \ , \\
h_{ij}^{0}\left(0\right) & = & \psi_{i}\left(0\right) +
\sum_{k\in\partial i\setminus j} R_{ki}(0)  \ , \label{eq:ms_h001}
\end{eqnarray}
the last equation corresponding to the unconstrained minimization over
the removal times of the neighbors of a vertex $i$ included in the
decycling set.

Let us give a more explicit interpretation of the quantities $L, R$ and of the last equations. $L_{ki}(t_i)$ is the minimum feasible cost in the subtree of $G\setminus i$ rooted in $k$ with the only condition that $t_k < t_i$ (see Eq.~\eqref{eq:ms_L1}). On the other hand, in Eq.~\eqref{eq:ms_R1} we define $R_{ki}(t_i)$ to be the minimum feasible cost in the subtree of $G\setminus i$ rooted in $k$ with $t_k\geq t_i$. As the message $h_{ij}^{1}\left(t_{i}\right)$ corresponds to a situation in which $j$ has already been removed at time $t_i$, one of the neighbors $k\in \partial i\setminus j$ can be removed after $i$. It follows that for $t_i>0$ the minimum feasible cost is given by the cost $\psi_{i}\left(t_{i}\right)$ plus the minimum between the minimum feasible cost when all neighbors $k\in \partial i\setminus j$ is removed before $i$  and the same quantity when one of the neighbors is allowed to be removed at a later time.

The field $h_i(t_i)$ is then obtained from the messages (see Eqs.~(\ref{eq:field1},\ref{eq:field2}))
according to 
\begin{eqnarray}
  h_{i}\left(t_{i}\right) & = & \psi_{i}\left(t_{i}\right) +
  \sum_{k\in\partial i} L_{ki}(t_i)+M_i\left(t_{i}\right) \ , \label{eq_field1bis}\\
  h_{i}\left(0\right) & = & \psi_{i}\left(0\right) +
  \sum_{k\in\partial i} R_{ki}(0)\ , \label{eq_field2bis} \\
M_i(t_i)&=&\min\{0,\min_{k\in\partial i}
\{R_{ki}(t_i)-L_{ki}(t_i)\}\} \ .
\end{eqnarray}

Finally a more efficient implementation can be devised, noting that
common quantities can be pre-computed in order to obtain the $h_{ij}$
for all the outgoing edges around a given vertex $i$. One indeed
obtains an implementation which runs in linear time both in $T$ and in
the degree $k_i$ of $i$ by defining
\begin{eqnarray}
S_i^0\left(t_{i}\right) & = & \sum_{k\in\partial i}L_{ki}\left(t_{i}\right) \ , \\
S_i^1 & = & \sum_{k\in\partial i}R_{ki}\left(0\right)\\
k_{i}\left(t_{i}\right) & \in & \arg\min_{k\in\partial i}\left\{ R_{ki}\left(t_{i}\right)-L_{ki}\left(t_{i}\right)\right\} \\
Q_{i}\left(t_{i}\right) & = & \min\left\{0,\min_{k\in\partial i\setminus k_{i}\left(t_{i}\right)}\left\{ R_{ki}\left(t_{i}\right)-L_{ki}\left(t_{i}\right)\right\}\right\} \\
M_{ij}\left(t_{i}\right) & = & \begin{cases}
  M_{i}\left(t_{i}\right)& \text{if} \ \ \ j \neq k_{i}\left(t_{i}\right)\\
  Q_{i}\left(t_{i}\right)& \text{if} \ \ \ \ensuremath{j=k_{i}\left(t_{i}\right)}
\end{cases}
\end{eqnarray}
which can all be computed in time $O\left(T k_i\right)$; we can then express the different values of the messages
$h_{ij}$ as 
\begin{eqnarray}
h_{ij}^{0}\left(t_{i}\right) & = & \psi_{i}\left(t_{i}\right)+S_i^0\left(t_{i}\right)-L_{ji}\left(t_{i}\right)\\
h_{ij}^{1}\left(t_{i}\right) & = &
\psi_{i}\left(t_{i}\right)+S_i^0\left(t_{i}\right)-L_{ji}\left(t_{i}\right)+M_{ij}\left(t_{i}\right) 
\\
h_{ij}^{0}\left(0\right) & = &
\psi_{i}\left(0\right)+S_i^1-R_{ji}\left(0\right)
\end{eqnarray}
which can be also computed in time $O\left(T\right)$ for each $j\in\partial i$.
The computation time for a complete iteration on all vertices $i$ is thus
$O\left(\sum_{i}k_{i}T\right)=O\left(\left|E\right|T\right)$. The
computation of the field $h_{i}\left(t_{i}\right)$ in \eqref{eq_field1bis}-\eqref{eq_field2bis}
is similar: 
\begin{eqnarray}
  h_{i}\left(t_{i}\right) & = &
  \psi_{i}\left(t_{i}\right)+S_i^0\left(t_{i}\right)+M_i\left(t_{i}\right)
  \qquad \text{for} \ \ 0<t_i\leq T \ ,\label{eq:field1b}\\
h_{i}\left(0\right) & = &
\psi_{i}\left(0\right)+S_i^1 \quad  \ .\label{eq:field2b}
\end{eqnarray}

\bigskip

This derivation of the Min-Sum equations shows that the algorithm is
exact on a tree: the recurrence equations on $h_{ij}$ are guaranteed
to converge, and the configuration $t_i^*$ obtained from the MS
expression of $h_i$ is the unique minimum of the cost function over
feasible configurations.
One can, however, iterate the recurrence equations on $h_{ij}$ for any
graph, and use the MS formalism as an heuristic algorithm that
provides a good approximation to the optimum, in particular when there
are not many short loops. There are, however, two issues with the
convergence of the message passing equations on $h_{ij}$:
\begin{itemize}
\item the $h_{ij}$ defined above are extensive energies, that would grow indefinitely in presence of loops in the graph. This problem is easily cured by adding a constant value $C_{ij}$ to all fields $h_{ij}(t_i,t_j)$, in such a way to keep the maximum entry of this matrix equal to a constant, for instance zero. This does not spoil the validity of the algorithm, as we only need informations about the relative energies of configurations to construct the decycling set: the optimum $t_i^*={\rm argmin} \, h_i(t_i)$ is obviously invariant by a shift of the reference energy.

\item even with this normalization the message passing equations are not guaranteed to converge. When they did not we enforced their convergence by employing a reinforcement procedure, that consists in taking $\psi_{i}\left(t_{i}\right)=\mathbb{I}\left[t_{i}=0\right]+\varepsilon_{i}\left(t_{i}\right)+\tau\gamma h_{i}\left(t_{i}\right)$
where $h_{i}$ is the local field computed with \eqref{eq:field1b}-\eqref{eq:field2b} in the previous iteration, $\gamma$ is a small real value
and $\tau$ is the iteration time. Typically we use in our simulation
$\gamma=10^{-3}$. 
\end{itemize}

\subsection{Tree-breaking in decycled graphs}
\label{sec_tree_breaking}

We explain now the second stage of our algorithm, namely the
dismantling of the acyclic graph obtained using MS in the first stage.

\subsubsection{Optimal tree breaking}

The computation of the $C$-dismantling number of a tree $G$ can be performed 
in a time growing polynomially with $C$ and with the size $N$ of the graph, by
the following dynamic programming approach.

Let us denote $G_{i\to j}$ the connected component of the vertex $i$
in the graph obtained from $G$ by removing one of its neighbors $j$, and
call $S_{ij}\left(c\right)$ the minimum number of vertices to be
removed from $G_{i\to j}$ to have that no component of the reduced graph is larger
than $C$ and that the component of $i$ is no larger than 
$c$. These quantities satisfy the following recursion: 

\begin{eqnarray*}
  S_{ij}\left(c\right) & = & \underset{\substack{\{c_k \}_{k\in\partial i\setminus j} \\ \underset{k\in\partial i\setminus j}{\sum} c_{k}\leq c-1}}{\min}
\sum_{k\in\partial i\setminus j} S_{ki}\left(c_{k}\right)\mbox{   if } 0< c\leq C \ , \\
  S_{ij}\left(0\right) & = & 1+\sum_{k\in\partial i\setminus j} S_{ki}\left(C\right) \ .
\end{eqnarray*}

Using max-convolutions (see e.g.~\cite{baldassi_max-sum_2015}) these quantities can be computed on all directed edges of the tree in time $O(N C^2)$. By adding an extra leaf $i'$ attached to a node $i$ on the tree, the quantity $S_{ii'}(C)$ gives the decycling number. A small modification can be used to also find optimal dismantling sets in time $O(N C^2)$. Even though polynomial, this complexity is often too expensive in practice even for moderate values of $C$. Fortunately we will see below a greedy strategy that achieves almost the same performance.

\subsubsection{Greedy tree-breaking}

An alternative approach to the dismantling of a forest is to follow a
greedy heuristic, removing iteratively the node in the largest
connected component  of the forest (i.e. a tree) that leaves the
smallest largest component. This procedure is guaranteed to
$C$-dismantle the forest by removing $S=N/C$ vertices or less. This
ensures the dismantling to a sublinear size of the largest component
$C\leq N/\log N$ by removing a sublinear number of vertices $S\leq
\log N$. Moreover, it can be implemented in time $O(N(\log N + T))$
where $T$ is the maximal diameter of the trees inside this forest. 
The worst case in terms of number of removed nodes is reached in the case of a one-dimensional chain, in which one needs to remove $S=2^k-1$ nodes to obtain components of size $C(S)\leq N/2^k$.

For a given tree $G$ on $N$ vertices, let us call $F$ the subset of vertices $i$ which are
optimal in the above sense, namely such that the removal of $i$ from
$G$ minimizes the size of the largest component of $G\setminus i$.
The elements of $F$ can be characterized in a very simple way. Denote
by $C(i)$ the size of the largest component of $G\setminus i$, so
$C(i)=\max_{j\in \partial i} |G_{j\to i}|$ and $F=\arg\min_{i\in
  V}C(i)$. Then $i^\star\in F$ if and only if $C(i^\star)\leq
N/2$. Suppose indeed that for $i^\star\in F$, $C(i^\star)>N/2$ and
take $j\in\partial i$ such that $|G_{j\to i^\star}|=C(i^\star)$. Then,
as $|G_{i^\star\to j}| + |G_{j\to i^\star}| = N$, we have that
$C(j)<\max\{N/2, C(i^\star)\}=C(i^\star)$ which is absurd. Conversely,
suppose that $C(i)\leq N/2$ and take $i^\star\in F\setminus i$. Consider the unique path $(i, k_1, \dots, k_n, i^\star)$ in $G$. Then $|G_{k_1\to i}|\leq C(i)$, and $|G_{i\to k_1}|\geq N-C(i)\geq N/2$. But $G_{i\to k_1}\subseteq G_{k_n\to i^\star}$ so $C(i^\star)\geq N/2$.

This characterization of $F$ can be used constructively to find an
$i^\star\in F$ efficiently. Pick for each connected component of the
initial forest a ``root'' vertex $i_0\in V$. For each $i$ compute
$w_{i}=|G_{i\to j}|$ where $j$ is the unique neighbor of $i$ on the
path between $i$ and the root $i_0$,
starting from the leaves and exploiting the relation
$w_{i}=1+\sum_{k\in \partial i\setminus j} w_{k}$;
note that $C(i_0)=\max_{j\in\partial i_0} w_{j}$.
 Place $i_0$ into a priority queue with priority given by the
 component size $K(i_0)=1+\sum_{j\in\partial i_0} w_j$. Iteratively
 pick the largest component from the queue. Then construct the sequence $i_t$ as follows: for every $t$, if $C(i_t)\leq N/2$, then $i^\star=i_t\in F$ and the process stops. Otherwise, iteratively choose $i_{t+1}$ such that $w_{i_{t+1}} = C(i_t) > N/2$. 

Once $i^\star$ is chosen and removed, the component is broken into $|\partial i^\star|$ components, each one rooted at $k\in\partial i^\star$. From these, only the component rooted at $i_{t-1}$ needs to have its $w$ values updated, as its orientation changed. The only needed adjustments are along the path $i_0,i_1,\dots,i_t$ and can computed in time proportional to $t$, which is bounded by the diameter of the tree, which is in turn bounded by $T$.

As the cost of the priority queue updates scale as $O(\log N)$, the
total number of operations for each vertex removal is thus $O(\log N +
T)$, hence the total number of operations for greedily dismantling a
forest scales as $O(N\cdot(\log N + T))$, as claimed above.

We performed an extensive comparison of the optimal and greedy
procedure for values of $C$ sufficiently small for the optimal one to
be doable in a reasonable time, using as a benchmark the forest output
by the MS algorithm applied to an Erd\"os-R\'enyi random graph of 78125
nodes and average degree 3.5. As shown in
Fig.~\ref{fig:optimal-greedy} we found the greedy strategy to have
very close to optimal performances, therefore we used this much faster
procedure in all other numerical simulations.

\begin{figure}
  \includegraphics[width=0.5\columnwidth]{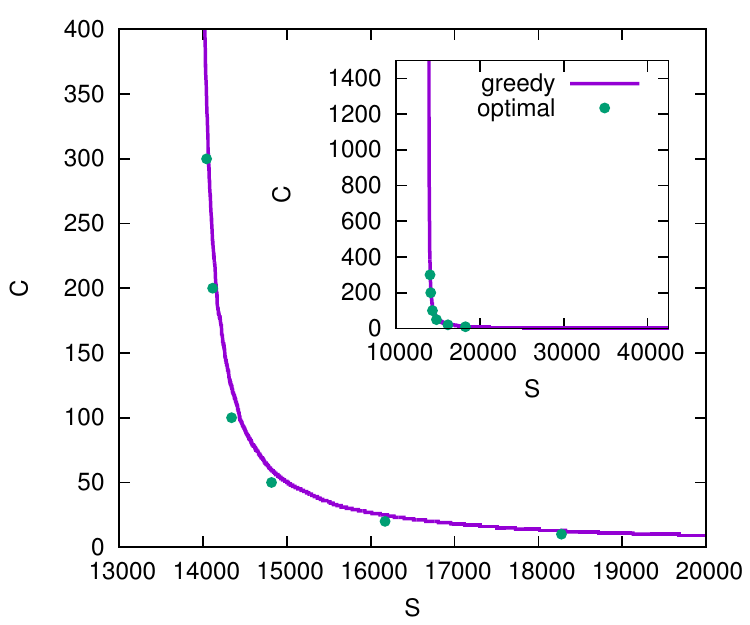}
  \caption{Comparison between greedy and optimal forest-breaking on
    the output of the MS algorithm for an Erd\"os-R\'enyi random graph
    with $78,125$ vertices and average degree $3.5$.\label{fig:optimal-greedy}}
\end{figure}

\subsection{Greedy reintroduction of cycles}
\label{sec_greedy_reintroduction}

The initial condition for the reverse greedy procedure is the graph obtained after the removal from $G$ of a set $S^{0}$ of nodes (dismantling set) and characterized by largest connected components of size $C$. Let us consider a target value $C'>C$ for the size of the largest connected components. 
As long as the size of the largest connected components in the graph is smaller than $C'$, the removed nodes are reintroduced one at a time by means of the following greedy strategy: at each iteration step $t$, we choose for reinsertion the node $i \in S^{t}$ (and the edges to
vertices in $V\setminus S^{t}$) such that the connected component $V_i^{t}$ the node $i$ ends up in is the smallest possible. 
An efficient implementation of the greedy reinsertion is easily obtained by maintaining a priority queue of the 
removed vertices with priority given by the size $|V_i^{t}|$. When a vertex $i$ is reintroduced in the graph, the number of connected components that get merged is at most equal to the degree $k_i$ of the vertex (in the original graph $G$). The number of elements in the priority queue that have to be modified after the reinsertion of $i$ is bounded by the number of nodes $j \in S^{t}$ that are connected to the new component in the original graph $G$. As the size of the largest component is at most $C'$, this number is at most $k_{\rm max} C'$, where $k_{\rm max}$ is the maximal degree of the graph. The computational cost of reintroducing a vertex in the graph is entirely given by the one of updating the queue, which is thus bounded by $k_{\rm max} C' \log (k_{\rm max} C')$. In a sparse graph, the update cost is thus typically sublinear in $N$, making the reverse greedy strategy very efficient.

\subsection{Competing algorithms}
\label{sec_otheralgo}
\subsubsection{Simulated Annealing} 

Besides our main algorithm based on the Min-Sum procedure we have
studied the network dismantling problem using simulated
annealing, i.e. building a Monte Carlo Markov Chain that makes a
random walk in the space of configurations of the subsets $S \subset V$
of removed vertices. We assign an energy to each configuration
according to 
\begin{equation}
\mathcal{E} = |S| \nu + C,
\end{equation}
in which $|S|$ is the number of removed nodes, $C$ the size of the
largest connected component in the graph obtained by removing $S$, and
$\nu$ is a free parameter. Note indeed that a set $S$ of removed
vertices can be considered ``good'' for two reasons: either because it
is small, or because its removal fragments the graph into small
components. These two figures of merits obviously contradict each
other and cannot be optimized simultaneously, $\nu$ thus controls the
balance between these two frustrating goals.

As usual in simulated annealing algorithms we introduce an inverse
temperature $\beta$ that is slowly increased during the evolution of
the Markov Chain, and at each time step $t$ we consider a move from the
current configuration $S^t$ to a new configuration $S_{\rm new}$, that
is accepted according to a standard Metropolis criterion, i.e. with probability 
\begin{equation}
\min{\left[1, e^{-\beta\left[ \mathcal{E}_{\rm new} -  \mathcal{E}^t
      \right]}\right]} \ ,
\end{equation}
where $\mathcal{E}_{\rm new}$ and $\mathcal{E}^t$ are the energies of
$S_{\rm new}$ and $S^t$ respectively. If the move is accepted we set
$S^{t+1}=S_{\rm new}$, $\mathcal{E}^{t+1}=\mathcal{E}_{\rm new}$,
otherwise the Markov Chain remains in the same configuration.
The proposed configuration $S_{\rm new}$ is constructed in the
following way at each time step: a node $i$ is chosen uniformly at
random among all the $N$ vertices of the graph, and its status is
reversed (if $i \in S^t$ then $S_{\rm new}=S^t \setminus i$, if $i
\notin S^t$ then $S_{\rm new}=S^t \cup \{i\}$). We then need to
compute the energy $\mathcal{E}_{\rm new}$ of this proposed
configuration; the first term is easily dealt with as $|S|$ varies by
$\pm 1$ depending on whether $i \in S^t$ or not. We thus only need to
compute the size $C_{\rm new}$ of the largest component in the new
configuration, facing three possible cases:
\begin{enumerate}
\item[1.] if $i \notin S^t$ and $i$ belongs to the largest
  component of $G \setminus S^t$ the size $C_{\rm new}$ of the largest component is recomputed;
\item[2.] if $i \notin S^t$ but $i$ does not belong to the largest
  component of $G \setminus S^t$ then $C_{\rm new}=C^t$ does not change;
\item[3.] if  $i\in S^t$ then it is only necessary to compute the size $C_i$ of the cluster $i$ belongs to once it is reintroduced in the graph and compare the latter with the current largest component, i.e. $C_{\rm new} = \max(C^t,C_i)$.
\end{enumerate}

The Markov chain is irreducible, recurrent and aperiodic, thus ergodic
and the Metropolis criterion ensures detailed balance, therefore
the SA algorithm would sample correctly the probability measure
$\propto e^{-\beta \mathcal{E}}$ if run with an infinitesimally small annealing
velocity. Unlike standard applications of simulated annealing, such as spin systems with short-range interactions, in the
present problem a single move (node removal/reintroduction) can
produce energy variations over a large range of scales, with the
consequence that there is no natural criterion to choose the annealing
protocol. We tested several different
annealing protocols and we adopted one in which the inverse
temperature $\beta$ is increased linearly from $\beta_{\rm min}$
to $\beta_{\rm max}$ (thus concentrating the measure on close to
ground-state configurations), with an increment of $d\beta$ at each time step (i.e. after each one attempted move). Protocols in which the inverse temperature is varied only after $O(N)$ attempted moves were also considered, with no relevant difference in the results. Similarly, there is no natural choice of the initial conditions. We tested the cases in which the initial set $S^0$ is empty and in which nodes are randomly assigned to $S^0$ independently with probability $1/2$, but for sufficiently small values of $\beta_{\rm min}$, different choices had no relevant effects on the optimization process.  

Fig.~\ref{fig:MCscaling} displays the minimum energy achieved using
the SA algorithm (with $\nu=0.6$) on Erd\"{o}s-R\'enyi random graphs
of average degree $d=3.5$ and increasing sizes from $N=1024$ to
$N=16384$. For comparison we also plot the results obtained using the
Min-Sum algorithm (horizontal lines). For small sizes, the SA
algorithm outperforms Min-Sum when the annealing scheme is
sufficiently slow ($d\beta$ very small). 
 Increasing $N$, the quality of the results obtained with SA degrades,
 as it would require an increasingly slower annealing protocol in
 order to achieve the same results obtained using Min-Sum. These
 results show that, even though the SA implementation proposed is
 simple and relatively fast even on large networks, the necessity of
 an increasingly slower annealing protocol prevents SA from reaching
 optimal results in a reasonable computational time.   

The results for simulated annealing presented in Fig.~1 and Fig.~4 of the main text are obtained 
 with parameters $d\beta=10^{-8}$, $\beta_{\rm min}=0.5$ for both, $\beta_{\rm max} = 20,$ $\nu= 1.2$ for Fig.~1, and $\beta_{\rm max} =10$, $\nu=2.0$ for Fig.~4. 

\begin{figure}[tbh]
\begin{center}
\includegraphics[width=0.5\columnwidth]{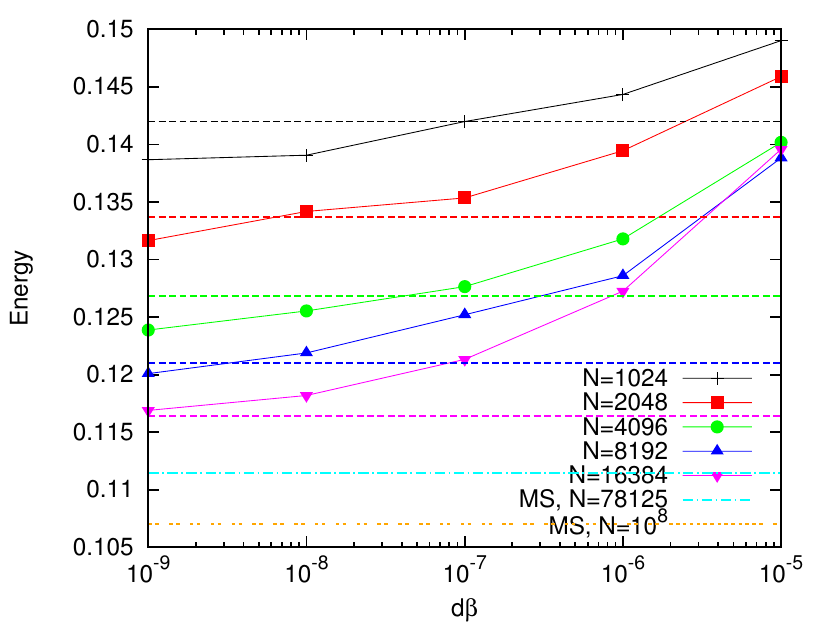}
\end{center}
\caption{Energy $\mathcal{E}$ as function of the rate $d\beta$ used to increase $\beta$ from $\beta_{\rm min}=0.5$ to $\beta_{\rm max}=30$ in the SA algorithm ($\nu=0.6$) used to dismantle Erd\"{o}s-R\'enyi random graphs of average degree $d=3.5$ and increasing size from $N=1024$ (top) to $N=16384$ (bottom). We also report the minimum energy achieved with the Min-Sum (MS) algorithm on the same graphs and for graphs of larger sizes $N=78125$ (cyan line) and $N=10^8$ (orange line). 
 \label{fig:MCscaling}} 
\end{figure}

\subsubsection{Score-based algorithms}

Let us give here more details about the other dismantling algorithms to which we
compared our own proposals. They all proceed by the (irreversible)
removal of nodes from the graphs to be dismantled, the differences
between them relying in the choice of a score function that assigns to
each vertex $i$ of the graph a score $e_i$, the vertices being removed
in the order of decreasing scores (with random choices in case of
ties). This quantity $e_i$ should be an
heuristic measure of the importance, or centrality, of the vertex $i$,
in the sense that more central nodes should lead to a larger decrease
in the size of the largest component when $i$ is removed. We have 
investigated the following score functions, the names corresponding to
the key in the figures:
\begin{itemize}
\item RND, $e_i=1$ for all $i$; this leads to a random choice of the
  removed vertices, i.e. to classical site percolation.
\item DEG, $e_i=k_i$ the degree of node $i$, this corresponds to removing
  the highest degree nodes first.
\item EC, for eigenvector centrality, uses as a score the
  eigenvector $e_i$ associated to
the largest eigenvalue $\lambda$ of the adjacency matrix
$A_{ij}=\mathbb{I}[\langle i,j \rangle \in E]$ of the graph, in other
words the solution of the linear system of equations 
\begin{equation}\label{ec}
\lambda \, e_i = \sum_{j\in \partial i} e_j \ .
\end{equation}
For a connected graph the
Perron-Frobenius theorem ensures that this eigenvector is unique and that
it can be chosen positive. 

\item ${\rm CI}_\ell$, for collective influence at level $\ell$, is a centrality measure introduced
  by Morone and Makse \cite{morone_influence_2015} to provide a
  heuristic measure of the influence that a node has on the neighbors
  within a certain distance $\ell$ from it. The collective influence
  of node $i$ at level $\ell$ is defined as
\begin{equation}
{\rm CI}_\ell(i) = (k_i-1) \sum_{j\in\partial B(i,\ell)} (k_j-1)
\end{equation}   
where $\partial B(i,\ell)$ denotes the set formed by all the nodes
that are at distance $\ell$ from node $i$ \cite{morone_influence_2015}. 
The CI value of node $i$ takes two contributions, the degree of node
$i$ and the number of edges emerging at distance $\ell$ from a ball
surrounding $i$. On expander graphs, such as random graphs, the number
of nodes contained in a ball $B(i,\ell)$ grows exponentially with
$\ell$, hence the calculation of the collective influence scores for
all nodes of the graph becomes computationally demanding already for
moderately small distance values ($\ell = 4,5$). 

\end{itemize}

We also made some tests with the score function defined as the
betweenness centrality and as the non-backtracking centrality~\cite{NewmanCentrality},
but for the graphs we considered the results we obtained were both
qualitatively and quantitatively similar (or worse) to those obtained
using EC, hence we do not report them.

For a given score function one can envision different ways to
implement the dismantling algorithm; the simplest would be to compute
the scores for all vertices of the original graph, and then to remove
the vertices in the order defined by this ranking. We used instead an
adaptive version, which gives much better results, that consist in
recomputing the scores of all remaining vertices after each removal of
the node with highest score in the current graph; all the results
presented in the main text and the Appendices have been obtained in this way. 
Even if it performs better
this adaptive strategy is also much more computationally demanding; an
intermediate compromise between these two extreme strategies would be
to recompute the scores only after a finite fraction $x$ of nodes is removed. 
Another implementation twist consists in recomputing the scores only
for the vertices belonging to the currently largest connected
components, as the removal of a vertex outside it would not decrease
the size of the largest component. This is useful in particular if one
tries to compute the EC scores by the power method (multiplying
several times an initial guess by the adjacency matrix); instead of
the full adjacency matrix one can consider only the submatrix
corresponding to the vertices in the largest component. By
construction this submatrix is irreducible and the power method will
converge, hence solving the possible convergence issues encountered by
the power method in the case of coexistence of several
connected components in the graph. The restriction to the largest
component modifies also the behavior of the DEG heuristic, as it
avoids the removal of large degree nodes in already small components.

\section{Other real-world and scale-free graphs}
\label{sec_othergraphs}

\begin{figure}[th]
\begin{center}
\includegraphics[width=0.5\columnwidth]{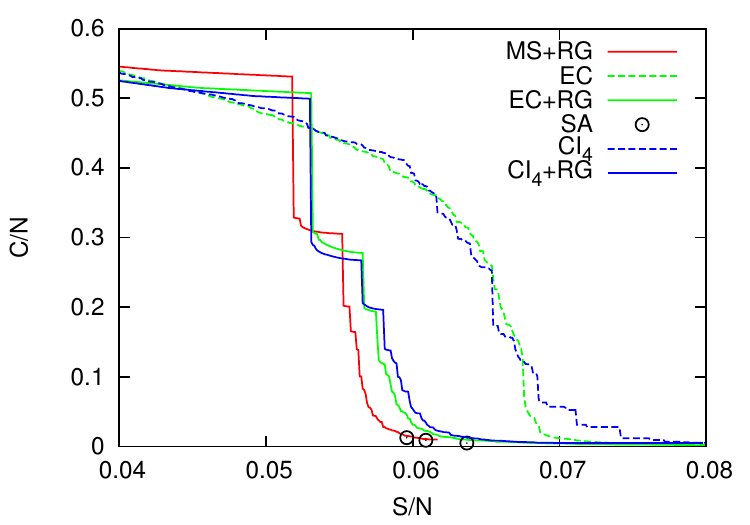}
\caption{Size of the largest component in scale-free random graphs, of
  size $N=10^4$ nodes and degree distribution $P(d) \propto
  d^{-\gamma}$ with $\gamma=2.5$, as function of the fraction of
  removed nodes $S$. The nodes are removed using adaptive Eigenvector
  Centrality (EC), adaptive Eigenvector Centrality plus Reverse Greedy
  (EC+RG),  Collective Influence with diameter $\ell=4$ (${\rm
    CI}_4$), the same plus Reverse Greedy (${\rm CI}_4$+RG), Min-Sum plus Reverse Greedy (MS+RG), and Simulated Annealing (SA) with $d\beta=10^{-8}$ and $\beta_{\rm min}=0.5, \beta_{\rm max}=20$ (and several values of $\nu$).
 \label{fig:SF}} 
\end{center}
\end{figure}

We already explained in the main
text that dismantling a graph by means of the decycling (plus greedy
tree breaking) is guaranteed to be optimal only for sparse random
graphs with locally tree-like structure. Nevertheless, we observed
that when the algorithm is complemented by a simple reverse greedy
(RG) strategy the final result is usually very good also on networks
in which many small loops are present, such as in the case of the
Twitter graph in Fig.~\ref{fig:Twitter} in the main text. Our way to state the quality
of the result is the direct comparison with the other available
algorithms, that are the Simulated Annealing algorithm and the other
heuristics (e.g. EC, CI) also complemented by the RG strategy.

\begin{figure}[thb]
\begin{center}
\includegraphics[width=0.5\columnwidth]{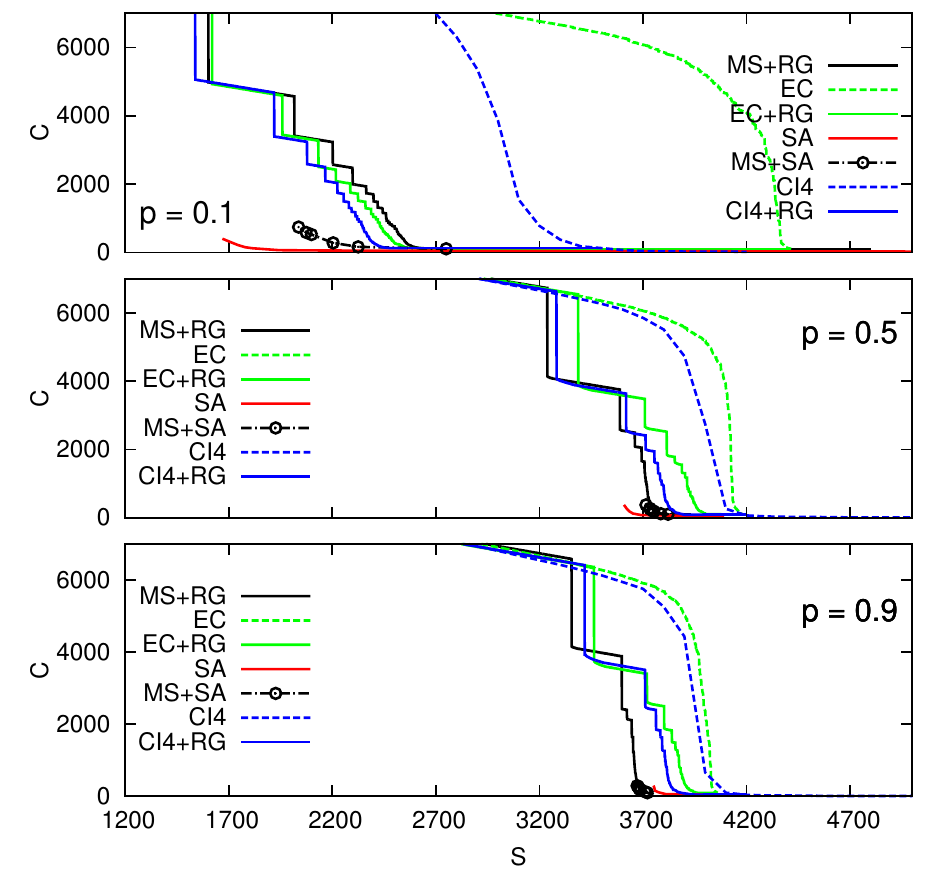}
\caption{Size of the largest component in WS small-world networks with rewiring probability $p=0.1,0.5,0.9$ achieved
  by removing a fraction of nodes using Eigenvector Centrality (EC), Eigenvector Centrality plus Reverse Greedy (EC+RG),  Collective Influence with diameter $\ell=4$ (CI4), Collective Influence with diameter $\ell=4$ plus Reverse Greedy (CI4+RG), Min-Sum (MS), Min-Sum plus Reverse Greedy (MS+RG), Min-Sum plus SA (MS+SA),  and Simulated Annealing (SA) with $d\beta=10^{-8}$ and $\beta_{\rm min}=0.5, \beta_{\rm max}=20$ (and several values of $\nu$).
 \label{fig:WS}} 
\end{center}
\end{figure}

We studied dismantling in the youtube network \cite{snapnets} with
$1.13$ million nodes and
concluded that the reverse greedy is of immense importance
here. Specifically we obtained that in order to dismantle the network
into components smaller that $C=1000$ nodes the CI methods removes $5.12\%$,
the ER removes $4.97\%$, the MS removes $5.67\%$ nodes. The reverse greedy
procedure improves all the these methods and gets dismantling sizes
$4.03\%$ for CI+RG, $4.07\%$ for EC+RG, and $3.97\%$ for MS+RG.

We also studied dismantling on an example of a synthetic scale-free
network. Results reported in Fig.~\ref{fig:SF}, are
qualitatively comparable to the ones for real networks.

In order to better quantify the effect of a large clustering coefficient on the different algorithmic methods under study, we
considered a well-known class of random graphs with tunable clustering
coefficient, the small-world network model introduced by Watts and Strogatz \cite{watts1998collective}. The WS network is generated starting from a one-dimensional lattice in which every node is connected with $d/2$ nearest-neighbors on both sides, then each edge $(i,j)$ with $i>j$ is rewired to a randomly chosen node $k\ne j$ with probability $p$. Fig.~\ref{fig:WS} shows the result of dismantling WS networks of size $N=10^4$, $d=6$ and rewiring probability $p=0.1,0.5,0.9$. For $p=0.9$ the WS network is topological similar to a random graph, with very small clustering coefficient,  because almost all edges have been rewired. On this network, Min-Sum plus reverse greedy outperforms centrality-based heuristics (EC+RG and CI+RG) and gives results that are comparable with the best obtained using SA. For $p=0.5$, MS+RG still gives a very good result, only slightly worse than SA. We also replaced the reverse greedy procedure with a reverse Monte Carlo method, in which a dismantling set is sought by performing the SA algorithm from the solution of the MS algorithm, by keeping only an optimal subset of the nodes already removed. The replacement of the reverse greedy procedure with a Monte Carlo based method gives improved results for both $p=0.9$ and $p=0.5$. We stress that this could be another useful strategy to improve heuristic results even in large networks, because the SA algorithm runs on a fraction of the original graph.

When $p$ is further decreased, the structure of the WS network significantly departs from  that of a random graph and short loops start to play a very important role, it is clear that we do not expect decycling to be a good strategy for dismantling in this regime. In this regime SA performs about 30\% better than any other algorithm, even though complemented with the reverse greedy strategy. When we perform SA from the solutions obtained using MS, the results are improved but still far from the best results obtained using SA alone. This is due to the fact that, in clustered networks, the dismantling set obtained by SA is not a subset of the dismantling set obtained using any other heuristic strategy, with an overlap that is usually small.











\end{document}